\documentclass[11pt]{article}
\usepackage[utf8]{inputenc}
\usepackage[T1]{fontenc}
\usepackage{fullpage}

\usepackage{amsthm}
\usepackage{amssymb,amsmath,mathtools}
\usepackage{graphicx,tabularx}
\usepackage{subcaption}
\usepackage{enumerate}
\usepackage{paralist}
\usepackage{xspace}
\usepackage{color}
\usepackage[boxed,vlined,ruled]{algorithm2e}
\SetAlCapSkip{0.5em}
\usepackage{booktabs}
\usepackage{varwidth}
\usepackage[nocompress]{cite}
\usepackage{hyperref}
\usepackage{authblk}

\newtheorem{lemma}{Lemma}

\newcommand{\Oh}{{\ensuremath{\mathcal{O}}}}
\newcommand{\df}{\textit}
\newcommand{\Q}{\ensuremath{\mathcal{Q}}\xspace}
\newcommand{\D}{\ensuremath{\mathcal{D}}\xspace}
\newcommand{\eps}{\ensuremath{\varepsilon}}

\newcommand{\NP}[0]{\texttt{NP}}

\newcommand{\rnd}{\operatorname{random}}
\newcommand{\fanout}{\operatorname{fanout}}
\newcommand{\pfanout}{\operatorname{p-fanout}}

\DeclareMathOperator*{\argmax}{arg\,max}

\title{Social Hash Partitioner:\\ A Scalable Distributed Hypergraph Partitioner}

\author[1]{Igor Kabiljo}
\author[1]{Brian Karrer}
\author[1]{Mayank Pundir}
\author[1]{Sergey Pupyrev}
\author[1]{Alon Shalita}
\affil[1]{Facebook}

\begin{document}
\date{}

\maketitle
\vspace{-1.5cm}
\begin{abstract}
    We design and implement a distributed algorithm for balanced $k$-way hypergraph partitioning that minimizes fanout, a fundamental hypergraph quantity also known as the communication volume and ($k-1$)-cut metric, by optimizing a novel objective called \textit{probabilistic fanout}.  This choice allows a simple local search heuristic to achieve comparable solution quality to the best existing hypergraph partitioners. 
    
    Our algorithm is arbitrarily scalable due to a careful design that controls computational complexity, space complexity, and communication.  In practice, we commonly process hypergraphs with billions of vertices and hyperedges in a few hours. We explain how the algorithm's scalability, both in terms of hypergraph size and bucket count, is limited only by the number of machines available.  We perform an extensive comparison to existing distributed hypergraph partitioners and find that our approach is able to optimize hypergraphs roughly $100$ times bigger on the same set of machines.
    
    We call the resulting tool \textit{Social Hash Partitioner} (SHP), and accompanying this paper, we open-source the most scalable version based on recursive bisection. 
\end{abstract}

\section{Introduction}
\label{sect:intro}
The goal of graph partitioning is to divide the vertices of a graph into a number of equal size components, so as
to minimize the number of edges that cross components. It is a classical and well-studied problem with origins in parallel scientific computing and VLSI design placement~\cite{BMSSS16}. Hypergraph partitioning is relatively less well-studied than graph partitioning.  Unlike a graph, in a hypergraph, an edge, referred to as a hyperedge, can connect to any number of vertices, as opposed to just two.  The revised goal is to divide the vertices of a hypergraph into a number of equal size components, while 
minimizing the number of components hyperedges span.

While graph partitioning is utilized in a variety of applications, hypergraph partitioning can be a more accurate model of many real-world problems~\cite{DBHBC06,CJZM10,CA99}. 
For example, it has been successfully applied for optimizing distributed
systems~\cite{GHKS14,KQDK14,SH16} as well as distributed scientific computation~\cite{DBHBC06,CA99}.  For another example, hypergraph partitioning accurately models the problem of minimizing the number of transactions in distributed data placement~\cite{CJZM10}.

Our primary motivation for studying hypergraph partitioning comes from the problem of storage sharding common in distributed databases. Consider a scenario with a large
dataset whose data records are distributed across several storage servers. A query to the database may 
consume several data records. If the data records are located on multiple servers, the query is answered by
sending requests to each server. Hence, the assignment of data records to servers
determines the number of requests needed to process a query; this number is often called the
\emph{fanout} of the query. Queries with low fanout can be answered more quickly, as there is less chance of
contacting a slow server~\cite{SH16,taleAtScale}. Thus, a common optimization is to choose an assignment 
of data records that collocates the data required by different queries.

\begin{figure}[t]
    \begin{minipage}[b]{0.98\textwidth}
        \begin{subfigure}[t]{.31\textwidth}
            \centering
            \includegraphics[page=1]{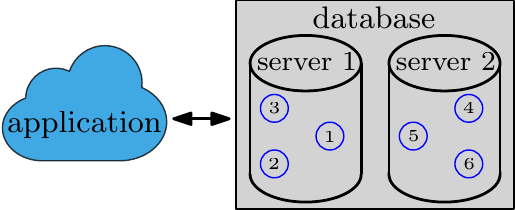}
            \caption{Storage sharding}
        \end{subfigure}
        \hfill
        \begin{subfigure}[t]{.31\textwidth}
            \centering
            \includegraphics[page=2]{figures/motivation}
            \caption{Model with bipartite graph}
            \label{fig:modelB}
        \end{subfigure}
        \hfill
        \begin{subfigure}[t]{.31\textwidth}
            \centering
            \includegraphics[page=3]{figures/motivation}
            \caption{Model with hypergraph}
            \label{fig:modelC}
        \end{subfigure}
        \caption{The storage sharding problem modeled with the (b)~bipartite graph and (c)~hypergraph 
            partitioning. Given three queries (red), $\{1,2,6\}$, $\{1,2,3,4\}$, $\{4,5,6\}$, the goal
            is to split six data vertices (blue) into two buckets, $V_1$ and $V_2$, so that the average
            query fanout is minimized. For $V_1=\{1,2,3\}$ and $V_2=\{4,5,6\}$, fanout of the queries
            is $2$, $2$, and $1$, respectively.}
        \label{fig:model}
    \end{minipage}
\end{figure}

Storage sharding is naturally modeled by the hypergraph partitioning problem; see Figure~\ref{fig:model}. 
For convenience, and entirely equivalent to the hypergraph formulation, we follow 
the notation of~\cite{GHKS14,DKKOPS16} and define the problem using a bipartite graph.
Let~$G=(\Q \cup \D, E)$ be an undirected bipartite
graph with disjoint sets of query vertices, $\Q$, and data vertices, $\D$.
The goal is to partition $\D$ into $k$ parts, that is, find a collection of $k$ disjoint
subsets $V_1, \dots , V_k$ covering $\D$ that minimizes an objective function. 
The resulting subsets, also called \df{buckets}, should be \emph{balanced}, that is,
$|V_i| \le (1+\eps) \frac{n}{k}$ for all $1 \le i \le k$ and some $\eps \ge 0$, where $n=|\D|$.

This optimization is exactly equivalent to balanced 
hypergraph partitioning with the set of vertices $\D$ and hyperedges $\{v_1,\dots,v_r\}$,
$v_i \in \D$ for every $q \in \Q$ with $\{q, v_i\} \in E$.  The data vertices are shared between the bipartite graph and its hypergraph representation, and each query vertex corresponds to a single hyperedge that spans all data vertices connected to that query vertex in the bipartite graph. Figures~\ref{fig:modelB} and~\ref{fig:modelC} show this equivalent representation.

For a given partitioning $P=\{V_1, \dots , V_k\}$ and a query vertex $q \in \Q$, we define
the \df{fanout} of $q$ as the number of distinct buckets having a data vertex incident to
$q$: $$\fanout(P, q) = |\{V_i : \exists \{q, v\} \in E, v \in V_i\}|.$$
The quality of partitioning $P$ is the average query fanout:
$$\fanout(P) = \frac{1}{|\Q|} \sum_{q \in \Q} \fanout(P, q).$$
The fanout minimization problem is, given a graph $G$, an integer $k > 1$, and a real 
number $\eps \ge 0$, find a partitioning of $G$ into $k$ subsets with the minimum average fanout.

Modulo irrelevant constant factors, the fanout
is also called the communication volume~\cite{DKUC15} and the ($k$ - 1) cut metric.  Arguably this is the most widely used objective for hypergraph partitioning~\cite{BMSW13,DKUC15}.\footnote{Although communication volume seems a more popular term for the measure
    in hypergraph partitioning community, we utilize fanout here to be consistent with database sharding terminology.}
Fanout is also closely related to the sum of external degrees, but not identical~\cite{KK00}.\footnote{The sum of external degrees (SOED) objective is equivalent to unnormalized fanout plus the number of query vertices with fanout greater than one (that is, SOED is the communication volume plus the hyperedge cut).}

Not surprisingly, the hypergraph partitioning problem is \NP-hard~\cite{AR06}. Hence, exact algorithms
are only capable of dealing with fairly small problem sizes.
Practical approaches all use heuristics. Even then, a significant problem with the majority of proposed algorithms is scalability.  Scalability is desirable because modern graphs can be massive; for example, the Facebook Social Graph contains billions of vertices and
trillions of edges, consuming many hundred of petabytes of storage space~\cite{SH16}.  

Scalability can be achieved through distributed computation, where the algorithm is distributed across
many workers that compute a solution in parallel. Earlier works on distributed hypergraph partitioning have proposed and implemented such algorithms, but as we argue in Section~\ref{sect:related} and
experimentally demonstrate in Section~\ref{sect:exp},
\emph{none of the existing tools are capable of solving the problem at large scale}.  No existing partitioner was able to partition a hypergraph with a billion hyperedges using four machines, and hypergraphs relevant to storage sharding at Facebook can be two or three orders of magnitude larger.  

With this in mind, \emph{we design, implement, and open source a scalable distributed algorithm for fanout minimization}.  We denote the resulting implementation \textit{Social Hash Partitioner} (\texttt{SHP}) because it can be used as the hypergraph partitioning component of the Social Hash framework introduced in~\cite{SH16}.  The greater framework has several other major components and the specific choice of hypergraph partitioner was only briefly mentioned previously, and this entirely self-contained paper delves into our partitioner in detail. The contributions of the paper are the following:

\begin{itemize}
    \item We design an iterative algorithm for hypergraph partitioning aimed at optimizing fanout through a classical local search heuristic~\cite{KL70} with two substantive modifications. The first is a novel objective function that generalizes fanout called \textit{probabilistic fanout}; using this objective for fanout minimization improves result quality and algorithm convergence. Our second modification facilitates a distributed implementation of the algorithm.
    
    \item We describe \texttt{SHP} and provide a detailed implementation of the algorithm that runs in a parallel manner, carefully explaining how it limits memory usage, computation, and communication.  \texttt{SHP} relies on the vertex-centric programming model and scales to hypergraphs with billions of vertices and hyperedges with a reasonable number of machines.  In addition to optimizing fanout, we also show that \texttt{SHP} can optimize other hypergraph objectives at scale.  We have open sourced the simple code for this implementation.  
    
    \item We present results of extensive experiments on a collection of real-world and synthetic hypergraphs. These results show that \texttt{SHP} is arbitrarily scalable while producing partitions of comparable quality to existing partitioners. \texttt{SHP} is able to partition hypergraphs $100x$ larger than existing distributed partitioners, making it the only solution capable of partitioning billion-vertex bipartite graphs using a cluster of machines.    
\end{itemize}

We emphasize that the paper describes an algorithmic solution utilized in our Social Hash framework. We refer the reader 
to~\cite{SH16} for the details of the framework, additional applications, and real-world experiments.


\section{Related Work}
\label{sect:related}

There exists a rich literature on graph partitioning from both theoretical and practical
points of view. We refer the reader to surveys by Bichot and Siarry~\cite{BS13} and 
by Bulu{\c{c}}~et~al.~\cite{BMSSS16}. Next we discuss existing approaches for storage sharding
that utilize graph partitioning and review theoretical results on the problem.
Then we describe existing libraries for hypergraph partitioning, focusing on the publicly 
available ones.
Finally, we analyze limitations of the tools, which motivate the development of
our new solution, \texttt{SHP}.

\paragraph*{Storage sharding.} 
Data partitioning is a core component of many existing distributed 
data management systems~\cite{CJZM10,GHKS14,KQDK14,SH16,Kie16}.
The basic idea is to co-locate related data so as to minimize communication overhead. 
In many systems, the data partitioning problem is reduced to a variant of graph 
partitioning~\cite{CJZM10,GHKS14,Kie16}; a solution is then found using a publicly available 
library such as \texttt{Metis}~\cite{KK95}. However, hypergraph partitioning is a better model
in several scenarios of storage sharding~\cite{KQDK14,SH16}. There exists much fewer
tools for partitioning of hypergraphs, as the problem tends to be harder and more 
computationally expensive than graph partitioning.

\paragraph*{Theoretical aspects.} 
The fanout minimization problem is a generalization of 
balanced $k$-way graph partitioning (also called minimum bisection when $k=2$), which is a central problem 
in design and analysis of approximation algorithms.
Andreev and R\"{a}cke~\cite{AR06} show that, unless \texttt{P}=\NP, there is no algorithm 
with a finite approximation factor for the problem when one requires a perfect balance, that is,
$\eps = 0$. Hence, most works focus on the case $\eps > 0$.
Leighton et al.~\cite{LMT90} and Simon and Teng~\cite{ST97} achieved
an $\Oh(\log k \log n)$ approximation algorithm for $\eps = 1$, that is, when the maximum 
size of the resulting buckets is $\frac{2n}{k}$.
The bound has been improved to $\Oh(\sqrt{\log k \log n})$ for $\eps=1$ by Krauthgamer et al.~\cite{KNS09} and
to $\Oh(\log n)$ for $\eps > 0$ by Feldmann and Foschini~\cite{FF15}.

We stress that being of high theoretical importance, the above approximation algorithms
are too slow to be used for large graphs, as they require solving linear programs. Hence,
most of the existing methods for graph and hypergraph partitioning are heuristics 
based on a simple local search optimization~\cite{BMSSS16}. We follow the same
direction and utilize a heuristic that can be implemented in an efficient way.
It is unlikely that one can provide strong theoretical guarantees on a local search
algorithm or improve existing bounds: Since fanout minimization is a generalization of
minimum bisection, it would imply a breakthrough result.

\paragraph*{Existing tools for hypergraph partitioning.}
As a generalization of graph partitioning, hypergraph 
partitioning is a more complicated topic and the corresponding algorithms are typically more 
compute and memory intensive. Here we focus on existing solutions designed for hypergraph partitioning.
\texttt{PaToH}~\cite{CA99}, \texttt{hMetis}~\cite{KAKS99}, and \texttt{Mondriaan}~\cite{VB05} 
are software packages providing single-machine algorithms. The tools
implement different variants of local refinement algorithms, such as Kernighan-Lin~\cite{KL70} or
Feduccia-Mattheyses~\cite{FM82}, that incrementally swap vertices among partitions to reduce
an optimization objective, until the process reaches a local minimum. Since such local search
algorithms can suffer from getting stuck in local minima, a multi-level paradigm is often used. The idea is to create a sequence
of ``coarsened'' hypergraphs that approximates the original hypergraph but have a smaller size. Then
the refinement heuristic is applied on the smaller hypergraph, and the process is reverted
by an uncoarsening procedure. Note that the above software packages all require random-access to the 
hypergraph located in memory; as a result these packages can only handle smaller hypergraphs.

\texttt{Parkway}~\cite{TK08} and \texttt{Zoltan}~\cite{DBHBC06} are distributed hypergraph partitioners that are based on a parallel version of the multi-level technique. Unfortunately, as we argue below and show in our experiments, the provided implementations 
do not scale well. 
We also mention a number of more recent tools for hypergraph partitioning. \texttt{UMPa}~\cite{DKUC15}
is a serial
partitioner with a novel refinement heuristic. It aims at minimizing several objective functions simultaneously. \texttt{HyperSwap}~\cite{YWMW16} is a distributed algorithm
that partitions hyperedges, rather than vertices. \texttt{rFM}~\cite{STA12}, allows 
replication of vertices in addition to vertex moves. We do not include these algorithms in our experimental 
study as they are not open-source and they are rather complex to be re-implemented in a fair way.

\paragraph*{Limitations of existing solutions.}
\texttt{Parkway} and \texttt{Zoltan} are two hypergraph partitioners that are designed 
to work in a distributed environment. Both of the tools implement a multi-level coarse/refine technique~\cite{KAKS99}. 
We analyzed the algorithms and identified the following scalability limitations.

\begin{itemize}
    \item First, multi-level schemes rely on an intermediate step in which the coarsest graph is partitioned on a single machine before it gets uncoarsened. While the approach is applicable for graph partitioning (when the coarsest graph is typically fairly small), 
    it does not always work for hypergraphs. For large instances, the number of \df{distinct} hyperedges can be substantial, and
    even the coarsest hypergraph might not fit the memory of a single machine.
        
    \item Second, the refinement phase itself is often equivalent to the local refinement scheme presented in our work, which if not done carefully can lead to scalability issues. For example, \texttt{Parkway} is using a single coordinator to approve vertex swaps while retaining balance. This coordinator holds the concrete lists of vertices and their desired movements, which leads to yet another single machine bottleneck.
    
    \item Third, the amount of communication messages between different machines/processors is an important aspect for a 
    distributed algorithm. Neither \texttt{Zoltan} nor \texttt{Parkway} provide strong guarantees on communication
    complexity. For example, the authors of \texttt{Zoltan} present their evaluation for mesh-like graphs (commonly used
    in scientific computing) and report relatively low communication overhead. Their results might not hold for 
    general non-planar graphs.
\end{itemize}    

In contrast, our algorithm is designed (as explained in the next section) to avoid these single machine bottlenecks and communication overheads.

\section{Social Hash Partitioner}
\label{sect:our_solution}

Our algorithm for the fanout minimization problem, as mentioned in Section~\ref{sect:intro}, assumes that the input is represented as a bipartite graph $G=(\Q \cup \D, E)$ with vertices representing queries and data 
objects, and edges representing which data objects are needed to answer the queries.
The input also specifies the number of servers that are available to serve the queries, $k>1$, and 
the allowed imbalance, $\eps > 0$.

\subsection{Algorithm}
\label{sect:algo}

For ease of presentation, we start with a high-level overview of our algorithm.  The basic idea is inspired by the Kernighan-Lin heuristic~\cite{KL70} for graph partitioning; see Algorithm~\ref{algo:bp}.
The algorithm begins with an initial random partitioning of data vertices into $k$ buckets. For every vertex, we independently 
pick a random bucket, which for large graphs guarantees an initial perfect balance. The algorithm then proceeds in multiple
steps by performing vertex swaps between the buckets in order to improve an objective function.
The process is repeated until a convergence criterion is met (e.g., no swapped vertices)
or the maximum number of iterations is reached.

Although the algorithm is similar to the classical one~\cite{KL70}, we introduce two critical modifications. The first concerns the objective function and is intended to improve
the quality of the final result.
The second is related to how we choose and swap vertices between buckets; this modification is needed to make distributed implementations efficient. 

\begin{algorithm}[t]
    \caption{Fanout Optimization}
    \label{algo:bp}
    
    \SetKwInOut{Input}{Input}
    \SetKwInOut{Output}{Output}
    \Input{graph $G=(\Q \cup \D, E)$, the number of buckets $k$, imbalance ratio $\eps$}
    \Output{buckets $V_1,V_2,\dots,V_k$}
    \BlankLine
    \For(\tcc*[f]{initial partitioning}){$v \in \D$}{
        $bucket[v] \leftarrow \rnd(1,k)$\;   
    }
    \Repeat(\tcc*[f]{local refinement}){converged {\bf or} iteration limit exceeded}{
        \For{$v \in \D$}{
            \For{$i=1$ \emph{\KwTo} $k$}{
                $gain_i[v] \leftarrow ComputeMoveGain(v, i)$\;
            }
            \tcc{find best bucket}
            $target[v] \leftarrow \argmax_i gain_i[v]$\;
            \tcc{update matrix}
            \If{$gain_{target[v]}[v] > 0$}{
                $S_{bucket[v],target[v]} \leftarrow S_{bucket[v],target[v]} + 1$\;
            }
        }
        \BlankLine
        \tcc{compute move probabilities}
        \For{$i,j=1$ \emph{\KwTo} $k$}{
            $probability[i,j] \leftarrow \frac{\min(S_{i,j}, S_{j,i})}{S_{i,j}}$\;
        }
        \BlankLine
        \tcc{change buckets}
        \For{$v \in \D$}{
            \If{$gains[v] > 0$ {\bf and} $\rnd(0,1) < probability[bucket[v],target[v]]$}{
                $bucket[v] \leftarrow target[v]$\;
            }
        }
    }
\end{algorithm}

\paragraph*{Optimization objective.}
Empirically, we found that fanout is rather hard to minimize with a local search heuristic. Such a heuristic can easily get stuck in a local minimum for fanout minimization.  
Figure~\ref{fig:fanout} illustrates this with an example which lacks a single move of a data vertex that improves fanout. All move gains are non-positive, and the local search algorithm stops in the non-optimal state.  To alleviate this problem, we propose a modified optimization objective and assume that a query $q \in Q$ only requires an adjacent data vertex $v \in \D, \{q, v\} \in E$ for some probability $p \in (0,1)$ fixed for all queries. This leads to so-called \df{probabilistic fanout}. The probabilistic fanout of a given query $q$, denoted by $\pfanout(q)$, is the expected number of servers that need to be contacted to answer the query given that each adjacent server needs to be contacted with independent probability $p$.

Formally, let $P=\{V_1, \dots , V_k\}$ be a given partitioning of $\D$, and let the number
of data vertices in bucket $V_i$ adjacent to query vertex $q$ be $n_i(q) = |\{v : v \in V_i \text{ and } \{q, v\} \in E\}|$.
Then server $i$ is expected to be queried with probability $1 - (1-p)^{n_i(q)}$.
Thus, the $\pfanout$ of $q$ is $\sum_{i=1}^k (1 - (1-p)^{n_i(q)})$, and our probabilistic fanout objective, denoted $\pfanout$, for Algorithm~\ref{algo:bp} is, given $p \in (0,1)$, minimize 
$$
\frac{1}{|Q|}\sum_{q \in \Q} \pfanout(q) = 
\frac{1}{|Q|}\sum_{q \in \Q} \sum_{i=1}^k \left(1 - (1-p)^{n_i(q)}\right).
$$

\begin{figure}[t]
    \centering
    \includegraphics[width=0.3\textwidth]{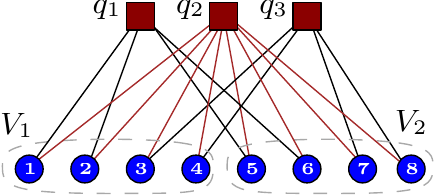}
    \caption{An example in which no single move/swap of data vertices improves query fanout.
        Probabilistic fanout (for every $0<p<1$) can be improved by exchanging buckets of vertices $4$ and $5$ 
        or by exchanging buckets of vertices $3$ and $6$. Applying both of the swaps reduces (non-probabilistic) fanout 
        of $q_1$ and $q_3$, which yields an optimal solution.}
    \label{fig:fanout}
\end{figure}

This revised objective is a smoother version of the (non-probabilistic) fanout.  It is simple to observe that $\pfanout(q)$ is less than or equal to $\fanout(q)$ for all $q \in Q$. If the
number of data adjacencies of $q$ in bucket $i$ is large enough, that is, $n_i(q) \gg 1$, then the bucket contributes to the objective 
a value close to $1$. If $n_i(q)=1$, then the contribution is $p$, which is smaller for $p<1$.
In contrast, the non-probabilistic fanout contribution is simply $1$, the same for all cases with $n_i(q) \ge 1$.

From a theoretical perspective, the way probabilistic fanout smooths fanout is by averaging the fanout objective over an ensemble of random graphs similar to the bipartite graph being partitioned.  A bipartite graph from this random graph ensemble is created by independently removing edges in the original bipartite graph with probability $p$.  Then the probabilistic fanout is precisely the expectation of fanout across this random graph ensemble.  In essence, $\pfanout$ minimization is forced to select a partition that performs robustly across a collection of similar hypergraphs, reducing the impact of local minima. With the new objective the state in Figure~\ref{fig:fanout} is no longer 
a local minimum, as data vertices $1$ and $2$ could be swapped to improve the $\pfanout$ of both $q_1$ and $q_3$.

An interesting aspect of $\pfanout$ is how it behaves in extreme cases. As we show next, when $p \to 1$, $\pfanout$ 
becomes the (non-probabilistic) fanout. When $p \to 0$, the new measure is equivalent to optimizing a weighted \df{edge-cut}, 
an objective suggested in prior literature~\cite{alpert1995recent, CA99, alpert1996hybrid}.
In practice it means that $\pfanout$ is a generalization of these measures and Algorithm~\ref{algo:bp} can be utilized to optimize either by setting small or large values of $p$.

\begin{lemma}
    \label{lm:p1}
    Minimizing $\pfanout$ in the limit as $p \to 1$ is equivalent to minimizing fanout.
\end{lemma}

\begin{proof}
    We write $\pfanout(q)$ for a query $q\in \Q$ as
    \begin{equation*}
    \sum_{i=1}^k \left(1 - (1-p)^{n_i(q)}\right) =
    \sum_{i=1}^k \left(1 - e^{n_i(q) \text{log}(1-p)} \right)
    \end{equation*}
    Now as $p \to 1$, $\text{log}(1-p)$ goes to negative infinity and the 
    exponential term is zero unless $n_i(q) = 0$ in which case it equals one. 
    Let $\delta(x) = 1$ if $x$ is true and $\delta(x)=0$, otherwise.
    In the limit, the above expression equals
    \begin{equation*}
    \sum_{i=1}^k \big( 1 - \delta(n_i(q) = 0) \big) =
    \sum_{i=1}^k \delta(n_i(q) > 0),
    \end{equation*}
    which is fanout of $q$.
\end{proof}

Next we show that optimizing $\pfanout$ as $p \to 0$ is equivalent to a graph partitioning problem
on an edge-weighted graph constructed from data vertices. 
For a pair of data vertices,
$u, v \in \D$, let $w(u, v)$ be the number of common queries shared by these data vertices, that is,
$w(u, v) = |\{q \in \Q : \{q, u\} \in E \text{ and } \{q, v\} \in E \}|$.
Consider a (complete) graph with vertex set $\D$ and let $w(u, v)$ be the weight of an edge between 
$u, v \in \D$.
For a given graph and a partition of its vertices, an \df{edge-cut} is the sum of edge weights between
vertices in different buckets.

\begin{lemma}
    \label{lm:p0}
    Minimizing $\pfanout$ in the limit as $p \to 0$ is equivalent to graph partitioning amongst the data vertices
    while minimizing weighted edge-cut, where the edge weight between $u \in \D$ and $v \in \D$ is given by $w(u, v)$.
\end{lemma}

\begin{proof}
    We begin from the definition of $\pfanout$ and consider the Taylor expansion
    around $p = 0$:
    
    \begin{alignat*}{2}
    \pfanout =& \sum_{q \in \Q} \sum_{i=1}^k \left(1 - (1-p)^{n_i(q)}\right) \nonumber \\
    =& \sum_{q \in \Q} \sum_{i=1}^k - \left( n_i(q)p + n_i(q)(n_i(q)-1) \frac{p^2}{2} + \Oh(p^3) \right) \nonumber\\
    =& C-\frac{p^2}{2} \sum_{q \in \Q} \sum_{i=1}^{k} n_i(q)^2 + \Oh(p^3)  \nonumber
    \end{alignat*}
    
    The first term, $C = \sum_{q \in \Q} \sum_{i=1}^{k} n_i(q)(\frac{p^2}{2}-p)$, is a constant proportional 
    to the number of edges in the graph. Thus, it is irrelevant to the minimization. The last
    term, $\Oh(p^3)$, can also be ignored for optimization when $p \to 0$.
    We simplify the second term further.
    
    \begin{eqnarray}
    -\frac{p^2}{2} \sum_{q \in \Q} \sum_{i=1}^{k} n_i(q)^2 &=& -\frac{p^2}{2} \sum_{i=1}^{k} \sum_{q \in \Q} n_i(q)^2 \nonumber \\
    &=& -\frac{p^2}{2} \sum_{i=1}^{k} \left( \sum_{u \in V_i} \sum_{v \in V_i} w(u, v) \right) \nonumber
    \end{eqnarray}
    
    Therefore, minimizing probabilistic fanout in the limit as $p \to 0$ is equivalent to maximizing the sum, taken 
    over all buckets, of edge weights between data vertices within the same bucket, or \df{maximizing within-bucket edge weights}.
    Alternatively, this is also equivalent to \df{minimizing intra-bucket edge weights}, that is,
    minimizing weighted edge-cut between buckets with edge weights given by $w(u, v)$.
\end{proof} 

This $p\rightarrow0$ limit is interesting because the resulting optimization is an instance of the clique-net model
suggested as a heuristic for hypergraph partitioning~\cite{alpert1995recent,CA99,alpert1996hybrid}. The idea is to convert
the hypergraph partitioning problem to the (simpler) graph partitioning problem. To this end, 
a hypergraph is transformed to an edge-weighted unipartite graph, $G^c$, on the same set of vertices, by adding a clique
amongst all pairs of vertices connected to a hyperedge. 
Multiple edges between a vertex pair in the resulting graph are combined by summing their respective weights.
The buckets produced by a graph partitioning algorithm on the new graph are then used as a solution
for hypergraph partitioning.

An obstacle for utilizing the clique-net model is the size of the resulting (unipartite) graph $G^c$. If there is a
hyperedge connecting $\Omega(n)$ vertices, then $G^c$ contains $\Omega(n^2)$ edges, even if the original hypergraph is 
sparse. Hence a common strategy is to use some variant of edge sampling to filter out edges with low weight in 
$G^c$~\cite{alpert1995recent,CA99,alpert1996hybrid}. Lemma~\ref{lm:p0} shows that this is unnecessary:
One can apply Algorithm~\ref{algo:bp} with a small value of fanout probability, $p$, for
solving the hypergraph partitioning problem in the clique-net model.

\paragraph*{Performing swaps.}
In order to iteratively swap vertices between buckets, we compute \df{move gains} for the vertices, that is,
the difference of the objective function after moving a vertex from its current bucket to another one.
For every vertex $v \in \D$, we compute $k$ values, referred to as $gain_i[v]$, indicating
the move gains to every bucket $1 \le i \le k$. Then every vertex chooses a
\df{target} bucket that corresponds to the highest move gain.  (For minimization, we select the bucket with the lowest move gain, or equivalently the highest negative move gain.)
This information is used to
calculate the number of vertices in bucket $i$ that chose target bucket $j$, denoted $S_{i,j}$, for all 
pairs of buckets $1 \le i, j \le k$. Ideally, we would move all these vertices from $i$ to $j$ to maximally
improve the objective.
However, to preserve balance across buckets, we exchange only $\min(S_{i,j}, S_{j,i})$ pairs of vertices between buckets $i$ and $j$. 

\begin{figure}[t]
    \centering
    \includegraphics[width=0.8\textwidth]{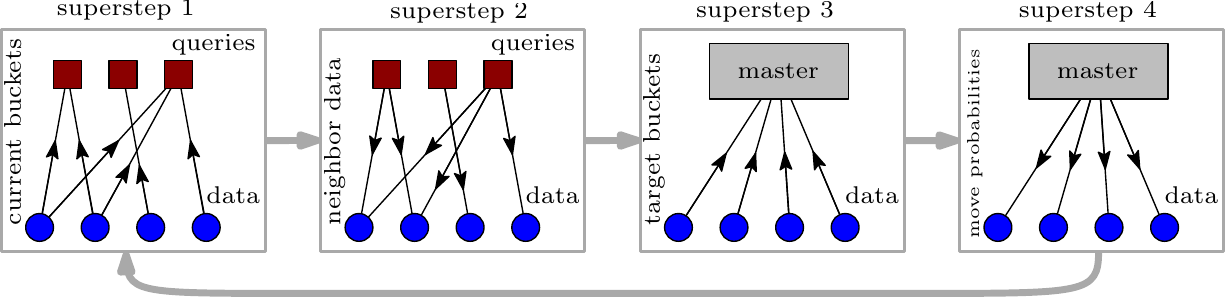}
    \caption{Distributed implementation of the fanout minimization algorithm in the vertex-centric framework
        with four supersteps and synchronization barriers between them:
        (1)~collecting query neighbor data, (2)~computing move gains, (3)~proposal of target buckets, 
        (4)~sending move probabilities and performing moves.}
    \label{fig:steps}
\end{figure}

In a distributed environment, we cannot pair vertices together for swaps easily, and so we instead perform the swaps in an approximate manner. Our implementation defines a probability for each vertex in bucket $i$ with
target bucket $j$ to be moved as $\frac{\min(S_{i,j}, S_{j,i})}{S_{i,j}}$. All vertices are then simultaneously moved to their target buckets respecting the computed probabilities, for all pairs of buckets.  With this choice of probabilities, the expected number of vertices moved from $i$ to $j$ and from $j$ to $i$ is the same; that is, the balance constraint is preserved in expectation. After that swap, we recompute move gains for all data vertices and proceed with the next iteration.

We compute the move gains as follows.  Assume that a vertex $v\in \D$ is moved from bucket $i$ to bucket $j$.
Our objective function, $\pfanout$, might change only for
queries adjacent to $v$ and for the terms corresponding to $V_i$ and $V_j$. Let $\mathcal{N}(v) \subseteq \Q$ be
the subset of queries adjacent to $v$. The move gain is then
\begin{equation}
\label{eq1}
\begin{split}
gain_j(v) & = \sum_{q \in \mathcal{N}(v)} \Big( \left(1-(1-p)^{n_i(q)-1}\right) + \left(1-(1-p)^{n_j(q)+1}\right) \Big) - \\
&  -  \sum_{q \in \mathcal{N}(v)} \Big( \left(1-(1-p)^{n_i(q)}\right) + \left(1-(1-p)^{n_j(q)} \right) \Big) = \\
& = \sum_{q \in \mathcal{N}(v)} \Big( -(1-p)^{n_i(q)-1} + (1-p)^{n_i(q)} -(1-p)^{n_j(q)+1} + (1-p)^{n_j(q)}\Big) \\
& = \sum_{q \in \mathcal{N}(v)} \Big( (1-p)^{n_i(q)-1} \cdot (-1+1-p) - (1-p)^{n_j(q)} \cdot (1-p+1) \Big) \\
& = p \cdot \sum_{q \in \mathcal{N}(v)} \Big((1-p)^{n_j(q)} - (1-p)^{n_i(q)-1}\Big).
\end{split}
\end{equation}

Next we provide details of the implementation.

\subsection{Implementation}
\label{sect:impl}
Our implementation relies on the vertex-centric programming model
and runs in the Giraph framework~\cite{giraph}. 
In Giraph, the input graph is stored as a collection of vertices that maintain
some local data (e.g., a list of adjacent vertices).
The vertices are distributed to multiple machines in a cluster and communicate
with each other via sending messages.
A computation in Giraph is split into supersteps that each consist
of the following processing steps: (i)~a vertex executes a user-defined function based on
local vertex data and on received messages, (ii)~the resulting output
is sent along outgoing edges. Note that since vertices operate only with local data,
such processing can easily be executed in parallel and in a distributed environment.
Supersteps end with a synchronization barrier, which
guarantees that messages sent in a given superstep are received at the beginning
of the next superstep. The whole computation is executed iteratively for a
certain number of rounds, or until a convergence property is met.

Algorithm~\ref{algo:bp} is implemented in the vertex-centric model in the following way; see Figure~\ref{fig:steps}.
The first two supersteps compute move gains for all data vertices. As can be seen from Equation~\ref{eq1},
a move gain of $v\in \D$ depends on the state of adjacent query vertices. Specifically, we need to know
the values $n_i(q)$ for every $q\in \Q$ and all buckets $1 \le i \le k$; we
call this information the \df{neighbor data} of query $q$. The first superstep
is used to collect the neighbor data; to this end, every data vertex sends its current bucket
to the adjacent queries, which aggregate the received messages into the neighbor data.
On the second superstep, the query vertices send their neighbor data back to adjacent
data vertices, which use this information to compute their move gains according to Equation~\ref{eq1}.

Once data vertices have computed move gains, we choose their target buckets and the next superstep aggregates this information in matrix $S$, that stores the number of candidate data vertices moving between pairs of buckets. The matrix is collected on a dedicated machine, called master, which computes move probabilities for the vertices. On the last superstep, the probabilities are propagated from master to all data vertices and the corresponding moves take effect.  This sequence of four supersteps continues until convergence.

\subsection{Complexity}
\label{sec:complexity}
Our primary consideration in designing the algorithm is to keep the implementation scalable to large instances. To this end, we limit
space, computational, and communication complexity such that each is bounded by $\Oh(k|E|)$, where $k$ is the number of buckets.

\paragraph*{Space complexity.}  We ensure every vertex $v$ of $G$ consumes only $\Oh(|\mathcal{N}(v)| + k)$
memory, where $|\mathcal{N}(v)|$ is the number of neighbors of vertex $v$, so the total memory consumption is $\Oh(|E| + k|V|)$. Every vertex keeps its adjacency list, which
is within the claimed limits. Additionally, every query vertex $q \in \Q$ stores its neighbor data containing 
$|fanout(q)| = \Oh(|\mathcal{N}(q)|)$ entries. Every data vertex $v\in \D$ stores move gains from its current
bucket to all other buckets, which includes up to $\Oh(k)$ values per vertex. The master machine stores
the information of move proposals for every pair of buckets, that is, its memory consumption is $\Oh(k^2) = \Oh(k|V|)$, which
is again within the claimed memory bound.  Notice that since Giraph distributes vertices among machines in a Giraph cluster randomly,
the scheme does not have a single memory bottleneck. All the machines are equally loaded, and in order to have enough memory for large graphs at a fixed number of buckets $k$, it is sufficient to increase the cluster size, while keeping the implementation unchanged.

\paragraph*{Computational complexity.} Our algorithm is within a bound of $\Oh(k|E|)$, assuming a constant number of refinement iterations. The computationally resource intensive steps are calculating the query neighbor data, which
is performed in $\Oh(|E|)$ time, and processing this information by data vertices. The latter step is bounded by
$\Oh(k|\mathcal{N}(v)|)$ steps for every $v\in \D$, as this is the amount of information being sent to $v$ in superstep 2.
Finally, computing and processing matrix $S$ requires $\Oh(|V|+k^2) = \Oh(k|V|)$ time.

\paragraph*{Communication complexity.}  Another important aspect of a distributed algorithm is its communication complexity, that is, the amount of messages
sent between machines during its execution. The supersteps $1$, $3$, and $4$ are relatively ``lightweight'' and require only
$|E|$, $|V|$, and $|V|$ messages of constant size, respectively. The ``heavy'' one is superstep~$2$ in
which every query vertex sends its neighbor data to all its neighbors. We can upper bound the amount of
sent information by $k|E|$, as neighbor data for every $q\in \Q$ contains $k$ entries. In practice, however, 
this is bounded by $\sum_{q \in \Q} fanout(q) \cdot |\mathcal{N}(q)|$, as the zero entries of neighbor data 
(having $n_i(q)=0$) need not be sent as messages. Hence, a reasonable estimation of the amount
of sent messages is $fanout \cdot |E|$ per iteration of Algorithm~\ref{algo:bp}, where $fanout$ is the average
fanout of queries on the current iteration.

We stress here that Giraph has several built-in optimizations that can further reduce the amount of sent and received messages.
For example, if exactly the same message is sent between a pair of machines several times (which might happen, for example, 
on superstep 2), then the messages are combined into a single one. Similarly, if a message is sent between 
two vertices residing on the same machine, then the message can be replaced with a read from the local memory of the machine.

Another straightforward optimization is to maintain some state of the vertices and only recompute the state between iterations 
of the algorithm when necessary. A natural candidate is the query neighbor data, which is recomputed only
when an adjacent data vertex is moved; if the data vertex stays in the same bucket on an iteration, then
it does not send messages on superstep 1 for the next iteration. Similarly, a data vertex $v \in \D$ may hold some state and recompute
move gains only in the case when necessary, that is, when another vertex $u \in \D$ with $\{q,v\} \in E$ and $\{q,u\} \in E$
for some $q \in \Q$ changes its bucket.

\paragraph*{Recursive partitioning.}
The discussion above suggests that Algorithm~\ref{algo:bp} is practical for small values of $k$, e.g., when $k=\Oh(1)$. In 
this case the complexity of the algorithm is linear in the size of the input hypergraph. However in some
applications, substantially more buckets are needed. In the extreme with $k=\Omega(|V|)$,
the implementation requires quadratic time and memory, and the run-time for large instances could be unreasonable even on a large Giraph cluster. 

A common solution to this problem in the hypergraph partitioning literature is to observe that the partitioning can be constructed recursively. In recursive partitioning, the algorithm splits data vertices into $r>1$ parts $V_1, \dots, V_r$. This splitting algorithm is recursively applied to all the graphs induced by vertices $\Q \cup V_1, \dots, \Q \cup V_r$ independently. The process continues
until we achieve $k$ buckets, which requires $\lceil \log_r k \rceil$ levels of recursion.
A typical strategy is to employ $r=2$; that is, recursive bisection~\cite{ST97,DKKOPS16}.

Notice that Algorithm~\ref{algo:bp} can be utilized for recursive partitioning with just a single modification.
At every recursion step, data vertices are constrained as to which buckets they are allowed to be moved to.
For example, at the first level of recursive bisection, all data vertices are split into $V_1$ and $V_2$.  At the second level, the vertices are split into four buckets $V_3, V_4, V_5, V_6$ so that the vertices $v \in V_1$ are
allowed to move between buckets $V_3$ and $V_4$, and the vertices $v \in V_2$ are allowed to moved between $V_5$ and $V_6$.
In general, the vertices of $V_i$ for $1 \le i \le k/2$ are split into $V_{2i+1}$ and $V_{2i+2}$.

An immediate implication of the constraint is that every 
data vertex only needs to compute $r$ move gains on each iteration.
Similarly, a query vertex needs to send only a subset of its neighbor data to data vertices that contains at most $r$ values.
Therefore, the memory requirement as well as the computational complexity of recursive partitioning is $\Oh(r|E|)$ per iteration, 
while the amount of messages sent on each iteration does not exceed $\Oh(r|E|)$, which is
a significant reduction over direct (non-recursive) partitioning when $r \ll k$. This improvement sometimes comes with the price of reduced quality; see Section~\ref{sect:exp} for a discussion. Accompanying the paper, we open-source a version that performs recursive bisection, as it is the most scalable.

\subsection{Advanced implementation}
\label{sect:details}
While the basic algorithm described above performs well, this subsection describes additional improvements that we have included in our implementation, motivated by our practical experience. 

First, randomly selecting vertices to swap between a pair of buckets $i$ and $j$ may not select those with the highest move gains to swap.  In the ideal serial implementation, we would have two queues of gains, one corresponding to vertices in bucket $i$ that want to move to $j$, and the other for vertices in bucket $j$ that want to move to $i$, sorted by move gain from highest to lowest.  We would then pair vertices off for swapping from highest to lowest.

This is difficult to implement exactly in a distributed environment.  Instead of maintaining two queues for each pair of buckets, we maintain two histograms that contain the number of vertices with move gains in exponentially sized bins.  We then match bins in the two histograms for maximal swapping with probability one, and then probabilistically pair the remaining vertices in the final matched bins.  In superstep 4, the master distributes move probabilities for each bin, most of which are either one or zero.  This change allows our implementation to focus on moving the most important gains first.  A further benefit is that we can allow non-positive move gains for target buckets.  A pair of positive and negative histogram bins can swap if the sum of the gains is expected to be positive, which frees up additional movement in the local search.

Additionally, we utilize the imbalance allowed by $\eps$ to consider imbalanced swaps.  For recursive partitioning, we typically do not want to allow $\varepsilon$ imbalance for the early recursive splits since that will substantially constrain movement at later steps of the recursion.  Instead, using $\varepsilon$ multiplied by the ratio of the current number of recursive splits to the final number of recursive splits works reasonably well in practice.  

Finally, when performing recursive partitioning, instead of optimizing $\pfanout$ for the current buckets, we approximately optimize for the final $\pfanout$ after all splits.  Consider a recursion step where an existing bucket will be eventually split into $t$ buckets.  If a query has $r$ neighbors in a particular bucket, this query-bucket pair gives a contribution to the current $\pfanout$ of $1 - (1-p)^r$.  We can (pessimistically) approximate the contribution of this pair to the final $\pfanout$ by assuming each of these $r$ neighbors has a probability of $1/t$ to end up in any of the $t$ final buckets.  Under this assumption, the contribution of this query to $\pfanout$ from a single final bucket is $(1-(1-p/t)^r)$, and summed across each of the $t$ buckets is $t \times (1 - (1 - p/t)^r$. 

\section{Experiments}
\label{sect:exp}
Here we describe our experiments that are designed to answer the following questions:
\begin{itemize}
    \item What is the impact of reducing fanout on query processing in a sharded database (Section~\ref{sec:ss})?
    
    \item How well does our algorithm perform in terms of quality and scalability compared to existing hypergraph partitioners
    (Sections~\ref{sec:qual} and \ref{sec:scal})?
    
    \item How do various parameters of our algorithm contribute to its performance and final result (Section~\ref{sec:params})?
    
\end{itemize}

\subsection{Datasets}

We use a collection of hypergraphs derived from large social networks and web graphs; see Table~\ref{table:dataset}.  We transform the input hypergraph into a bipartite graph, $G=(\Q \cup \D, E)$, as described in Section~\ref{sect:intro}. 

In addition, we use five large synthetically generated graphs that have similar characteristics as the Facebook friendship graph~\cite{ELWCK16}.
These generated graphs are a natural source of hypergraphs; in our storage sharding application, to render a profile-page of a Facebook user, one might want to fetch information about a user's friends.  Hence, every user of a social network serves both as query and as data in a bipartite graph.

In all experiments, isolated queries and queries of degree one (single-vertex hyperedges) are removed, since they do not contribute to the objective, having fanout equal to one in every partition.

\newcolumntype{R}{>{\raggedleft\arraybackslash}p{1.35cm}}

\begin{table}[!t]
    \centering
    \begin{tabular}{lrrrr}
        \toprule
        \centering hypergraph & source & \multicolumn{1}{c}{$|\Q|$} & \multicolumn{1}{c}{$|\D|$} & \multicolumn{1}{c}{$|E|$} \\
        \midrule
        \texttt{email-Enron}	  	& \cite{snap}  & $25,\!481$   & $36,\!692$   & $356,\!451$ \\
        \texttt{soc-Epinions} & \cite{snap} & $31,\!149$  & $75,\!879$   & $479,\!645$ \\
        \texttt{web-Stanford} & \cite{snap} & $253,\!097$  & $281,\!903$   & $2,\!283,\!863$ \\
        \texttt{web-BerkStan} & \cite{snap} & $609,\!527$  & $685,\!230$   & $7,\!529,\!636$ \\
        \texttt{soc-Pokec}    & \cite{snap}     & $1,\!277,\!002$  & $1,\!632,\!803$   & $30,\!466,\!873$ \\
        \texttt{soc-LJ} & \cite{snap}  & $3,\!392,\!317$  & $4,\!847,\!571$   & $68,\!077,\!638$ \\
        \texttt{FB-10M}	  	  & \cite{ELWCK16} & $32,\!296$ & $32,\!770$	 & $10,\!099,\!740$ \\
        \texttt{FB-50M}	  	  & \cite{ELWCK16} & $152,\!263$ & $154,\!551$	 & $49,\!998,\!426$ \\
        \texttt{FB-2B}	  	  & \cite{ELWCK16} & $6,\!063,\!442$ & $6,\!153,\!846$	 & $2 \times 10^9$ \\
        \texttt{FB-5B}	  	  & \cite{ELWCK16} & $15,\!150,\!402$ & $15,\!376,\!099$ & $5 \times 10^9$ \\
        \texttt{FB-10B}	  	  & \cite{ELWCK16} & $30,\!302,\!615$ & $40,\!361,\!708$ & $10 \times 10^9$ \\
        \bottomrule
    \end{tabular}
    \caption{Properties of hypergraphs used in our experiments.}
    \label{table:dataset}
\end{table}

\subsection{Evaluation}
We evaluate two versions of our algorithm, direct partitioning into $k$ buckets (\texttt{SHP-k}) and recursive
bisection with $\log_2 k$ levels (\texttt{SHP-2}). The algorithms are
implemented in Java and \texttt{SHP-2} is available at~\cite{giraphgit}.  Both versions can be run in a 
single-machine environment using one or several threads running in parallel, or in
a distributed environment using a Giraph cluster.
Throughout the section, we use imbalance ratio $\eps=0.05$.

\subsubsection{Storage Sharding}
\label{sec:ss}
Here we argue and experimentally demonstrate that fanout is a suitable objective function for our primary application, 
storage sharding. We refer 
the reader to the description of the Social Hash framework~\cite{SH16} for more detailed evaluation of the system and other 
applications of \texttt{SHP}.

In many applications, queries issue requests to multiple storage servers, and they do so in parallel.
As such, the latency of a multi-get query is determined by the slowest request. 
By reducing fanout, the probability of encountering a request that is unexpectedly slower than the others 
is reduced, thus reducing the latency of the query. 
This is the fundamental argument for using fanout as the objective function for the assignment problem 
in the context of storage sharding. We ran a simple experiment to confirm our understanding of the relationship 
between fanout and latency. We issued trivial remote requests and measured (i)~the latency of a single request 
($\fanout=1$) and (ii)~the latency of several requests sent in parallel ($\fanout > 1$) (that is,
the maximum over the latencies of single requests).
Figure~\ref{fig:lat_fan_1} shows the results of this experiment and
illustrates the dependency between various percentiles of multi-get query latency and fanout of the query.
The observed latencies match our expectations and indicate that reducing fanout is important
for database sharding; for example, one can almost half the average latency by reducing fanout from $40$ to $10$.

\begin{figure}[!t]
    \centering
    \begin{minipage}[b]{0.98\textwidth}
        
    \begin{subfigure}[t]{0.49\textwidth}
        \includegraphics[width=\columnwidth]{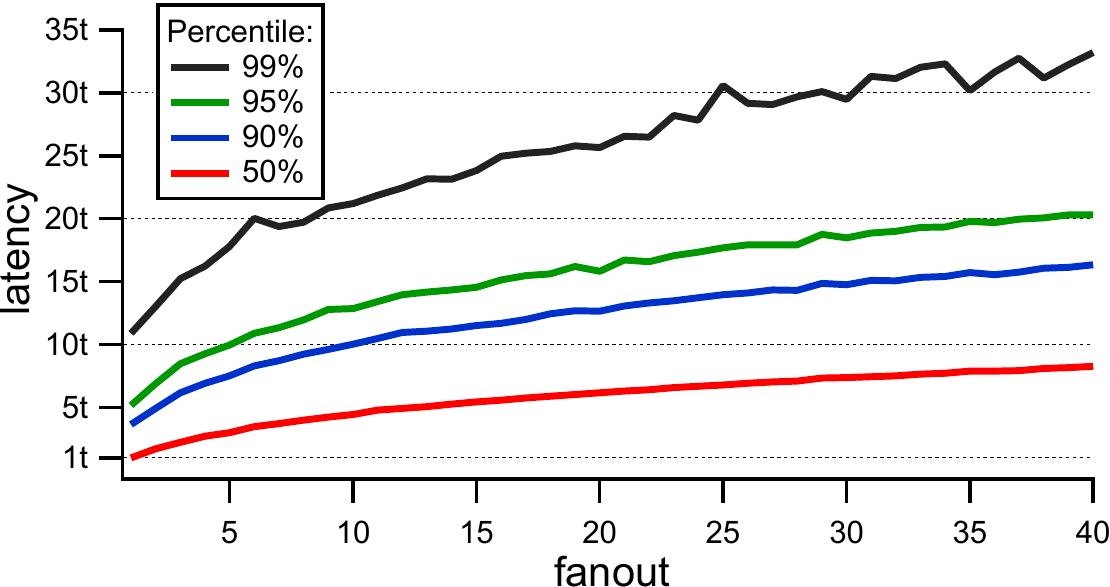}
        \caption{Synthetic queries}
        \label{fig:lat_fan_1}
    \end{subfigure}
    \hfill
    \begin{subfigure}[t]{0.49\textwidth}
        \includegraphics[width=\columnwidth]{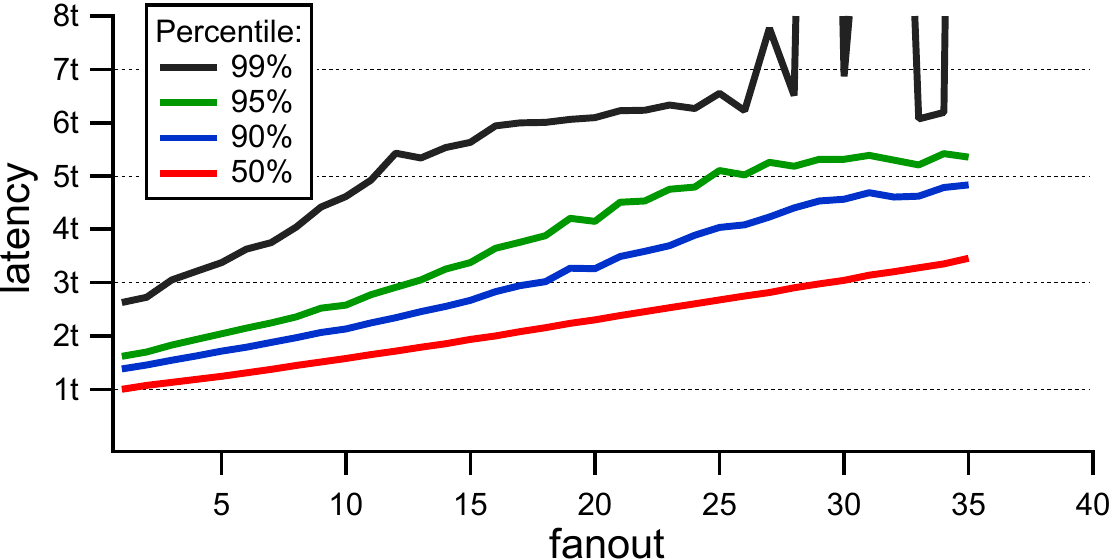}
        \caption{Real-world queries}
        \label{fig:lat_fan_2}
    \end{subfigure}
    \end{minipage}

    \caption{Distribution of latency for multi-get queries with various fanout, 
        where $t$ is the average latency of a single call.}
    \label{fig:multiget_latency_vs_fanout}
\end{figure}

There are several possible caveats to our analysis of the relationship between fanout and latency in the
simplistic experiment. For example, reducing fanout generally 
increases the size of the largest request to a server, which could increase latency.
With this in mind, we conduct a more realistic experiment with $40$ servers storing a subset of the Facebook friendship graph.
For the experiment, the data is stored in a memory-based, key-value store, and there is one data record per user.
In order to shard the data, we minimize fanout using our \texttt{SHP} algorithm applied for the graph.
We sample a live traffic pattern, and issued the same set of queries, while measuring fanout and latency of each
query. The dependency between fanout and latency are shown in Figure~\ref{fig:lat_fan_2}.
Here the queries needed to issue requests to only $9.9$ servers on average. Notice that we do not include 
measurements for $\fanout > 35$, as there are very few such queries; this also explain several ``bumps'' in latency for
queries with large fanout. The results demonstrate that decreasing fanout from $40$ (corresponding to a ``random'' sharding)
to $10$ (``social'' sharding) yields a $2x$ lower average latency for the queries, which
agrees with the results of the earlier experiment.

Finally, we mention that after we deployed storage sharding optimized with \texttt{SHP} to one of the graph databases at Facebook, containing thousands of storage servers, we found that measured latencies of queries decreased by over $50\%$ on average, and CPU utilization also decreased by over $50\%$; see \cite{SH16} for more details.

\subsubsection{Quality}
\label{sec:qual}
Next we compare the quality of our algorithm as measured by the optimized
fanout with existing hypergraph partitioners. We identified the following
tools for hypergraph partitioning whose implementations are publicly available:
\texttt{hMetis}~\cite{KAKS99}, \texttt{Mondriaan}~\cite{VB05}, 
\texttt{Parkway}~\cite{TK08},
\texttt{PaToH}~\cite{CA99}, 
and \texttt{Zoltan}~\cite{DBHBC06}. These are tools that can process hypergraphs and can optimize fanout (or the closely related sum of external degrees) as the objective function.
Unfortunately, one of the best (according to the recent DIMACS Implementation Challenge~\cite{BMSW13}) single-machine hypergraph packages, UMPa~\cite{DKUC15}, is not publicly available, 
and hence we do not include it in the evaluation.

For a fair comparison, we set allowed imbalance $\eps=0.05$ and used default optimization flags for all partitioners.  We also require that partitions can be
computed with $10$ hours without errors.  We computed the fanout of partitions produced with these algorithms for various numbers of buckets, $k$, and found that \texttt{hMetis} and \texttt{PaToH} generally produced higher fanout than the other partitioners on our hypergraphs.  So for clarity, we focus on results from \texttt{SHP-2} and \texttt{SHP-k}, along with \texttt{Mondriaan}, \texttt{Parkway}, and \texttt{Zoltan}.

\begin{table}[!ht]
    \begin{minipage}[t][][b]{0.5\linewidth}
        \vspace{0.5cm}
        \includegraphics[width=0.99\textwidth]{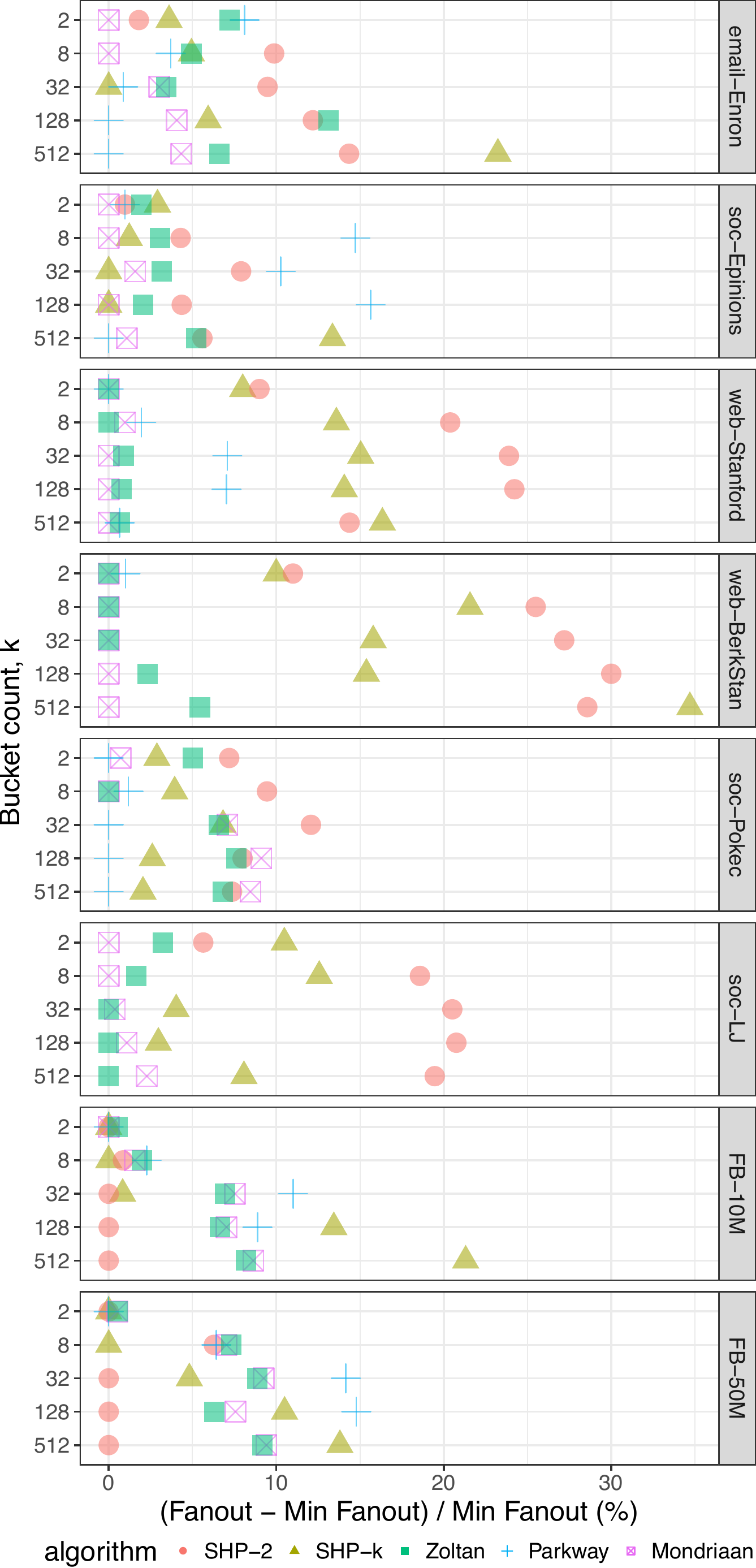}
    \end{minipage}
    \hfill    
    \begin{varwidth}[t][][b]{0.5\linewidth}
        \scriptsize
        \begin{tabular}{lrrrrr}
            \toprule
            algorithm
            & \multicolumn{1}{c}{$2$} & \multicolumn{1}{c}{$8$} & \multicolumn{1}{c}{$32$} & \multicolumn{1}{c}{$128$} & \multicolumn{1}{c}{$512$}\\
            \midrule
            \texttt{SHP-k} & $1.15$ & $1.7$ & $2.32$ & $3.39$ & $5.41$ \\
            \texttt{SHP-2} & $1.13$ & $1.78$ & $2.54$ & $3.59$ & $5.02$ \\
            \texttt{Mondriaan} & $1.11$ & $1.62$ & $2.39$ & $3.33$ & $4.58$ \\
            \texttt{Parkway} & $1.20$ & $1.68$ & $2.34$ & $3.20$ & $4.39$ \\
            \texttt{Zoltan} & $1.19$ & $1.7$ & $2.40$ & $3.62$ & $4.68$ \\[0.12cm]
            \midrule
            \texttt{SHP-k} & $1.06$ & $1.65$ & $2.53$ & $3.90$ & $6.28$ \\
            \texttt{SHP-2} & $1.04$ & $1.70$ & $2.73$ & $4.07$ & $5.85$ \\
            \texttt{Mondriaan} & $1.03$ & $1.63$ & $2.57$ & $3.90$ & $5.60$ \\
            \texttt{Parkway} & $1.04$ & $1.87$ & $2.79$ & $4.51$ & $5.54$ \\
            \texttt{Zoltan} & $1.05$ & $1.68$ & $2.61$ & $3.98$ & $5.83$ \\[0.12cm]
            \midrule
            \texttt{SHP-k} & $1.08$ & $1.17$ & $1.30$ & $1.46$ & $1.78$ \\
            \texttt{SHP-2} & $1.09$ & $1.24$ & $1.40$ & $1.59$ & $1.75$ \\
            \texttt{Mondriaan} & $1.01$ & $1.04$ & $1.13$ & $1.28$ & $1.53$ \\
            \texttt{Parkway} & $1.01$ & $1.05$ & $1.21$ & $1.37$ & $1.54$ \\
            \texttt{Zoltan} & $1.01$ & $1.03$ & $1.14$ & $1.29$ & $1.54$ \\[0.12cm]
            \midrule
            \texttt{SHP-k} & $1.10$ & $1.24$ & $1.32$ & $1.50$ & $1.98$ \\
            \texttt{SHP-2} & $1.11$ & $1.28$ & $1.45$ & $1.69$ & $1.89$ \\
            \texttt{Mondriaan} & $1.00$ & $1.02$ & $1.14$ & $1.30$ & $1.47$ \\
            \texttt{Parkway} & $1.01$ & & & & \\
            \texttt{Zoltan} & $1.00$ & $1.02$ & $1.14$ & $1.33$ & $1.55$ \\[0.12cm]
            \midrule
            \texttt{SHP-k} & $1.43$ & $2.64$ & $4.07$ & $5.52$ & $7.48$ \\
            \texttt{SHP-2} & $1.49$ & $2.78$ & $4.27$ & $5.81$ & $7.87$ \\
            \texttt{Mondriaan} & $1.40$ & $2.54$ & $4.08$ & $5.87$ & $7.95$ \\
            \texttt{Parkway} & $1.39$ & $2.57$ & $3.81$ & $5.38$ & $7.33$ \\
            \texttt{Zoltan} & $1.46$ & $2.54$ & $4.06$ & $5.79$ & $7.83$ \\[0.12cm]
            \midrule
            \texttt{SHP-k} & $1.37$ & $2.06$ & $2.84$ & $3.82$ & $5.22$ \\
            \texttt{SHP-2} & $1.31$ & $2.17$ & $3.29$ & $4.48$ & $5.77$ \\
            \texttt{Mondriaan} & $1.24$ & $1.83$ & $2.74$ & $3.75$ & $4.94$ \\
            \texttt{Parkway} & & & & & \\
            \texttt{Zoltan} & $1.28$ & $1.86$ & $2.73$ & $3.71$ & $4.83$ \\[0.12cm]
            \midrule
            \texttt{SHP-k} & $1.93$ & $7.04$ & $21.81$ & $61.45$ & $125.7$ \\
            \texttt{SHP-2} & $1.93$ & $7.10$ & $21.62$ & $54.17$ & $103.62$ \\
            \texttt{Mondriaan} & $1.93$ & $7.15$ & $23.25$ & $57.98$ & $112.56$ \\
            \texttt{Parkway} & $1.93$ & $7.20$ & $24.02$ & $58.98$ & \\
            \texttt{Zoltan} & $1.94$ & $7.18$ & $23.12$ & $57.77$ & $112.10$ \\[0.12cm]
            \midrule
            \texttt{SHP-k} & $1.93$ & $6.69$ & $22.43$ & $59.88$ & $117.16$ \\
            \texttt{SHP-2} & $1.93$ & $7.11$ & $21.40$ & $54.19$ & $102.95$ \\
            \texttt{Mondriaan} & $1.94$ & $7.16$ & $23.38$ & $58.29$ & $112.64$ \\
            \texttt{Parkway} & $1.93$ & $7.12$ & $24.43$ & $62.20$ &  \\
            \texttt{Zoltan} & $1.94$ & $7.18$ & $23.30$ & $57.61$ & $112.4$ \\[0.12cm]
            \bottomrule
        \end{tabular}
    \end{varwidth}%
    \caption{Fanout optimization of different single-machine partitioners across hypergraphs from 
    Table~\ref{table:dataset} for $k \in \{2, 8, 32, 128, 512\}$: (left)~relative quality over the lowest computed fanout, 
    and (right)~raw values of fanout. The shown results are the ones computed on a single machine within $10$ hours without failures.}
    \label{fig:quality_min_comparison}
\end{table}

Figure~\ref{fig:quality_min_comparison} compares the fanout of partitions produced by these
algorithms for various bucket count against the minimum fanout partition produced across all algorithms.  No partitioner is consistently the best across all hypergraphs and bucket count.  However, \texttt{Zoltan} and \texttt{Mondriaan} generally produce high quality partitions in all circumstances.  

\texttt{SHP-2}, \texttt{SHP-k} and to a lesser extent \texttt{Parkway} are more inconsistent.  For example, both versions of \texttt{SHP} have around a $10-30\%$ percentage increased fanout, depending on the bucket count, over the minimum fanout on the web hypergraphs.  On the other hand, \texttt{SHP} generally performs well on \texttt{FB-10M} and \texttt{FB-50M}, likely the hypergraphs closest to our storage sharding application. 

We also note a trade-off between quality and speed for recursive bisection compared to direct $k$-way optimization. The fanout from \texttt{SHP-2} is typically, but not always, $5-10\%$ larger than \texttt{SHP-k}. 

Because no partitioner consistently provides the lowest fanout, we conclude that using all partitioners and taking the lowest fanout across all partitions is an attractive option if possible.  However, as we shall see in the next section, \texttt{SHP-2} may be the only option for truly large hypergraphs due to its superior scalability.

\subsubsection{Scalability}
\label{sec:scal}
In this section, we evaluate \texttt{SHP}'s performance in a distributed setting using the four largest hypergraphs in Table~\ref{table:dataset}. We use 4 machines for our experiments each having the same configuration: Intel(R) Xeon(R) CPU E5-2660~@~2.20GHz with 144GB RAM.  

First, we numerically verify the computational complexity estimates from Section~\ref{sec:complexity}.  
Figure~\ref{fig:cpu_time} shows \texttt{SHP-2}'s total wall time as a function of the number of hyperedges across 
the four largest hypergraphs.
Notice that the y-axis is on a log-scale in the figure. The data verifies that \texttt{SHP-2}'s computational 
complexity is  $\Oh(\log k |E|)$, as predicted.

Next we analyze scalability of our approach with various number of worker machines in a cluster.
Figure~\ref{fig:cpu_time2}(left) illustrates the run-time of \texttt{SHP-2} using $4$, $8$, and $16$
machines; Figure~\ref{fig:cpu_time2}(right) illustrates the total wall time using the same configuration.
While there is a reduction in the run-time, the speedup is not proportional to the ratio
of added machines. We explain this by an increased communication between the machines, which
contributes to the performance of our algorithm.

\begin{figure}[!t]
    \centering
    
    \begin{minipage}[b]{\textwidth}
        
    \begin{subfigure}[t]{0.44\textwidth}
        \includegraphics[width=\columnwidth]{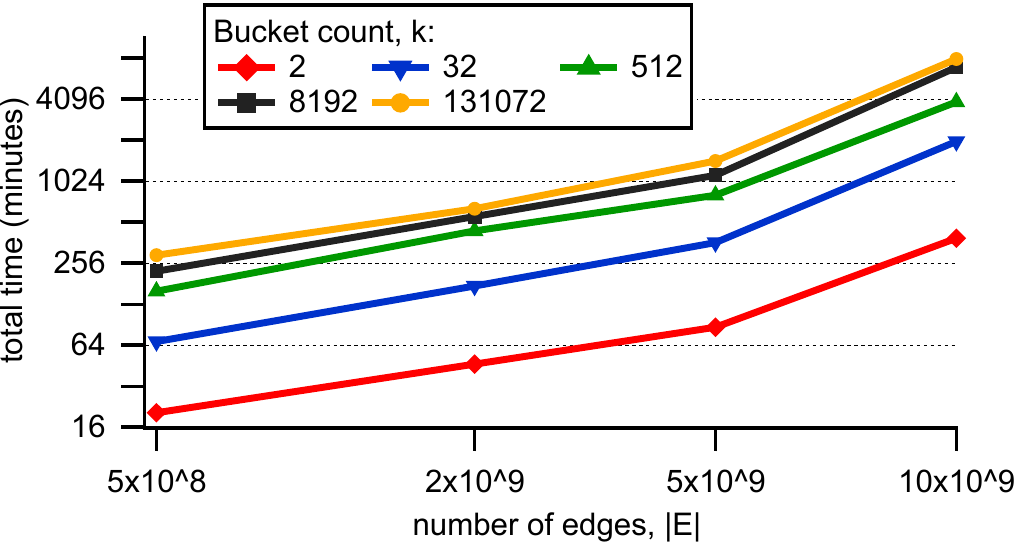}    
        \caption{Total time in a cluster with $4$ machines}
        \label{fig:cpu_time}
    \end{subfigure}
    \hfill
    \begin{subfigure}[t]{0.54\textwidth}
        \includegraphics[width=\columnwidth]{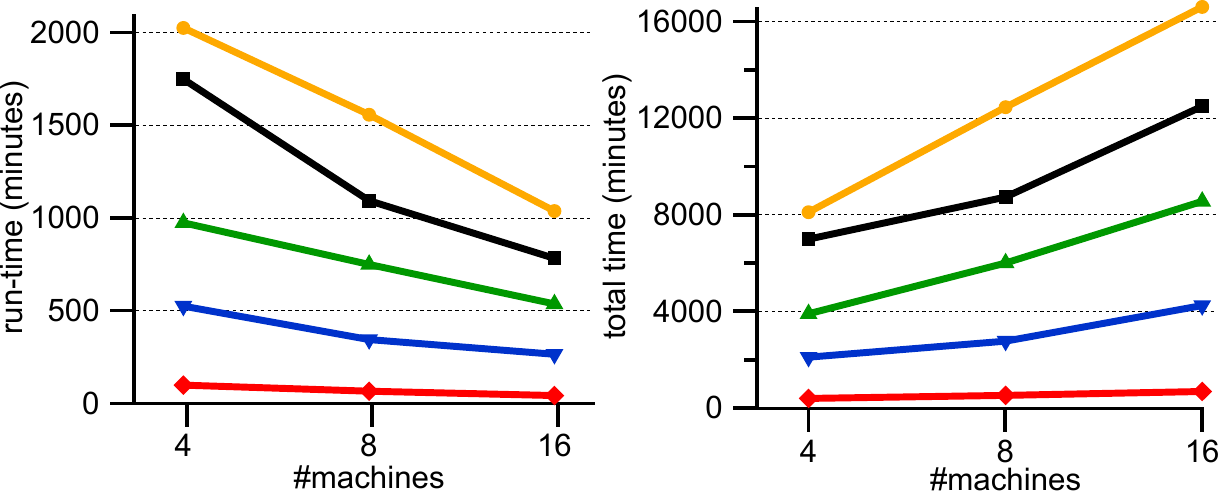}    
        \caption{Run-time and total time in a cluster with various
            number of machines on $\texttt{FB-10B}$}
        \label{fig:cpu_time2}
    \end{subfigure}
    
    \end{minipage}
        
    \caption{Scalability of \texttt{SHP-2} in a distributed setting for $k \in \{2, 32, 512, 8192, 131072\}$
        across largest hypergraphs from Table~\ref{table:dataset}.        
        The run-time is the time of processing a hypergraph in a cluster; the total time is the
        processing time of a single machines multiplied by the number of machines.}
\end{figure}

Now we compare SHP's performance to the two existing distributed partitioning packages (\texttt{Parkway} and \texttt{Zoltan}).  The comparison with \texttt{Zoltan} is particularly relevant since it provides generally lower fanout in the quality comparison.

Figure~\ref{fig:cpu_time_comparison} shows the run-time of these partitioners in minutes across the hypergraphs and various bucket count.  
If the result was not computed within $10$ hours without errors, we display the maximum value in the figure.  
Parkway only successfully ran on one of these graphs within the time allotted, because it runs out of memory on the other hypergraphs in the $4$-machine setting.  
Similarly, \texttt{Zoltan} also failed to partition hypergraphs larger than \texttt{soc-LJ}.  
On the other hand, \texttt{SHP-k} ran on \texttt{FB-10B} for $32$ buckets, and only \texttt{SHP-2} was able to successfully run on all tests.

\begin{table}[!ht]
    \begin{minipage}[t][][b]{0.5\linewidth}
        \vspace{0.7cm}
        \includegraphics[width=0.99\textwidth]{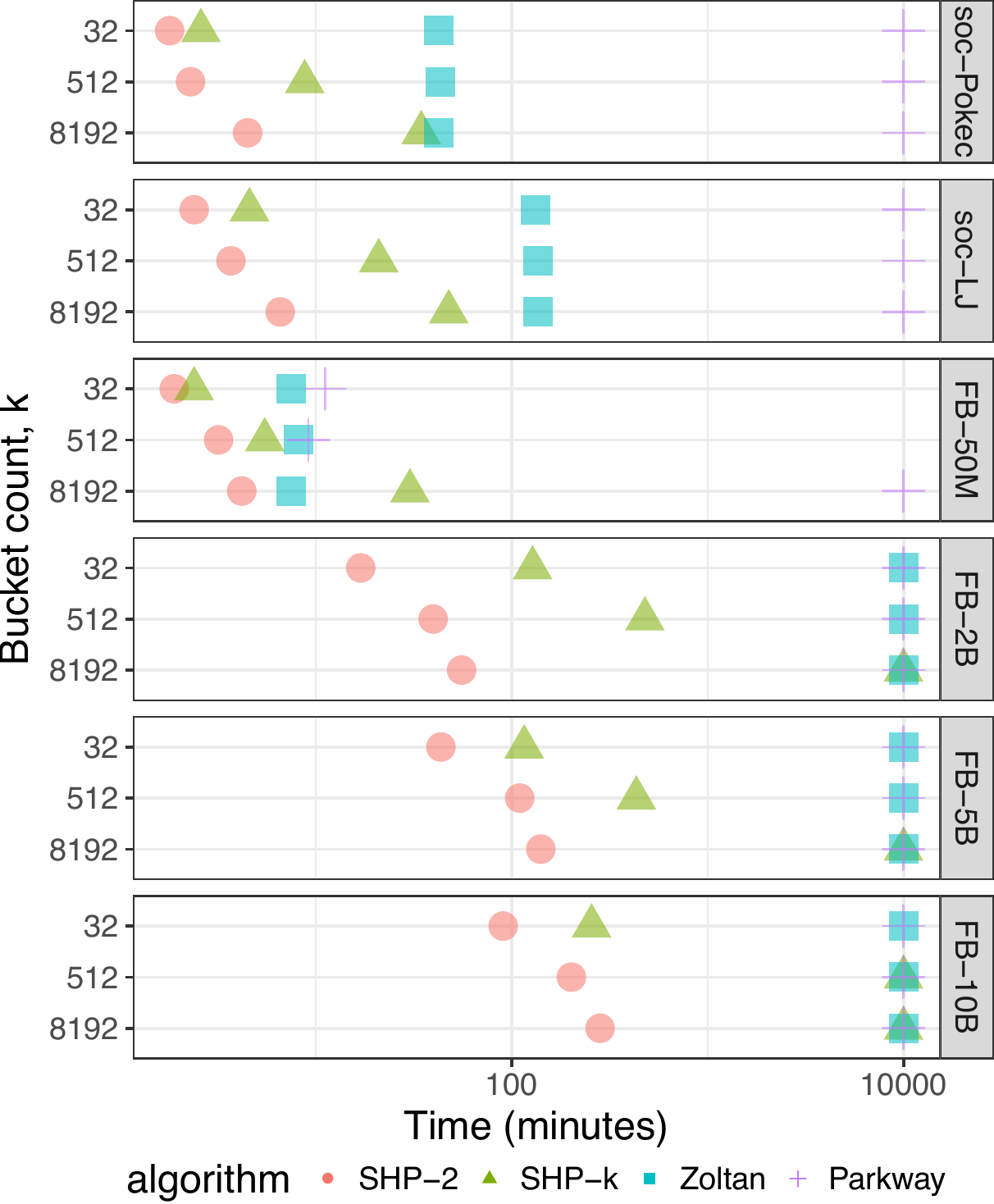}
    \end{minipage}
    \hfill    
    \begin{varwidth}[t][][b]{0.5\linewidth}
        \scriptsize
        \begin{tabular}{l@{\qquad}@{\qquad}r@{\qquad}@{\qquad}r@{\qquad}@{\qquad}r}
            \toprule
            algorithm
            & \multicolumn{1}{c}{$32$} & \multicolumn{1}{c}{$512$} & \multicolumn{1}{c}{$8192$} \\
            \midrule
            \texttt{SHP-k} & $2.6$ & $8.8$ & $34.6$ \\
            \texttt{SHP-2} & $1.8$ & $2.3$ & $4.5$ \\
            \texttt{Parkway} & & & \\
            \texttt{Zoltan} & $42.7$ & $43.4$ & $42.6$ \\
            \midrule
            \texttt{SHP-k} & $4.6$ & $21$ & $47.8$ \\
            \texttt{SHP-2} & $2.4$ & $3.7$ & $6.6$ \\
            \texttt{Parkway} & & & \\
            \texttt{Zoltan} & $133$ & $136$ & $136$ \\
            \midrule
            \texttt{SHP-k} & $2.4$ & $5.5$ & $30.3$ \\
            \texttt{SHP-2} & $1.9$ & $3.2$ & $4.2$ \\
            \texttt{Parkway} & $11.2$ & $9.21$ & \\
            \texttt{Zoltan} & $7.5$ & $8.2$ & $7.5$ \\
            \midrule
            \texttt{SHP-k} & $128$ & $479$ & \\
            \texttt{SHP-2} & $17$ & $39.8$ & $55.6$ \\
            \texttt{Parkway} & & & \\
            \texttt{Zoltan} & & & \\
            \midrule
            \texttt{SHP-k} & $116$ & $433$ & \\
            \texttt{SHP-2} & $43.6$ & $110$ & $141$ \\
            \texttt{Parkway} & & & \\
            \texttt{Zoltan} & & & \\
            \midrule
            \texttt{SHP-k} & $256$ & & \\
            \texttt{SHP-2} & $90.6$ & $202$ & $283$ \\
            \texttt{Parkway} & & & \\
            \texttt{Zoltan} & & & \\
            \bottomrule
        \end{tabular}
    \end{varwidth}%
    \caption{Run-time in minutes of distributed hypergraph partitioners across hypergraphs from Table~\ref{table:dataset} for 
    $k \in \{32, 512, 8192\}$: (left)~visual representation, and (right)~raw values.
     All tests run on 4 machines. Partitioners that fail process an instance or if their run-time exceeds $10$ hours are shown with the maximum value.}
    \label{fig:cpu_time_comparison}
\end{table}

Further, note that the x-axis is on a log-scale in Figure~\ref{fig:cpu_time_comparison}. So \texttt{SHP} can not only run on larger graphs with more buckets than \texttt{Zoltan} and \texttt{Parkway} on these $4$ machines, the run-time is generally substantially faster.  \texttt{SHP-2} finished partitioning every hypergraph in less than $5$ hours, and for the hypergraphs on which \texttt{SHP-k} succeeded, it ran less than $8$ hours.

While not observable in Figure~\ref{fig:cpu_time_comparison}, \texttt{Zoltan}'s run-time was largely independent of the bucket count, such that for $8192$ buckets on \texttt{FB-50M} it was faster than \texttt{SHP-k}.  This is a relatively rare case, and typically \texttt{SHP-k}, despite having a run-time that scales linearly with bucket count, was faster in our experiments.  While in all examples \texttt{Zoltan} was much slower than \texttt{SHP-2}, for a division of a small hypergraph into a very large number of buckets, \texttt{Zoltan} could conceivably be faster, since \texttt{SHP-2}'s run-time scales logarithmically with bucket count. 

\subsubsection{Parameters of \texttt{SHP}}
\label{sec:params}
There are two  parameters affecting Algorithm~\ref{algo:bp}: the fanout probability and
the number of refinement iterations. To investigate the effect of these parameters, 
we apply \texttt{SHP-2} for various values of $0 < p < 1$; see Figure~\ref{fig:prob} illustrating the resulting percentage reduction in
(non-probabilistic) fanout on \texttt{soc-Pokec}. Values between $0.4 \le p \le 0.8$ tend to produce the lowest fanout, with $p=0.5$ being
a reasonable default choice for all tested values of bucket count $k$. The point $p=1$ in the figure corresponds to optimizing fanout directly with \texttt{SHP-2}, and yields worse results than $p=0.5$. 

\begin{figure}[!t] 
    \centering
    \includegraphics[width=0.52\textwidth]{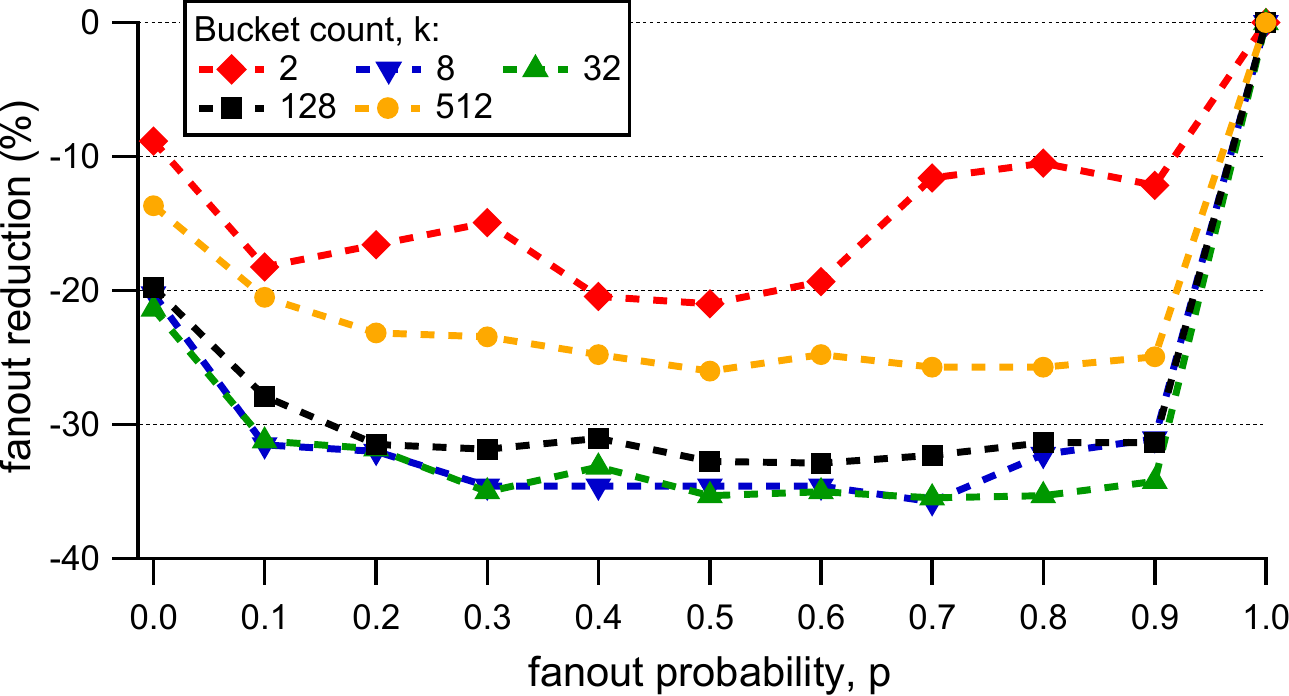}
    \caption{Fanout optimization with \texttt{SHP-2} on \texttt{soc-Pokec} as a function of 
        fanout probability, $p$, for different bucket counts.}
    \label{fig:prob}
\end{figure}    

\begin{figure}[!t] 
    \centering
    \begin{minipage}[b]{0.99\textwidth}
        \begin{subfigure}[t]{0.49\textwidth}
            \includegraphics[width=\columnwidth]{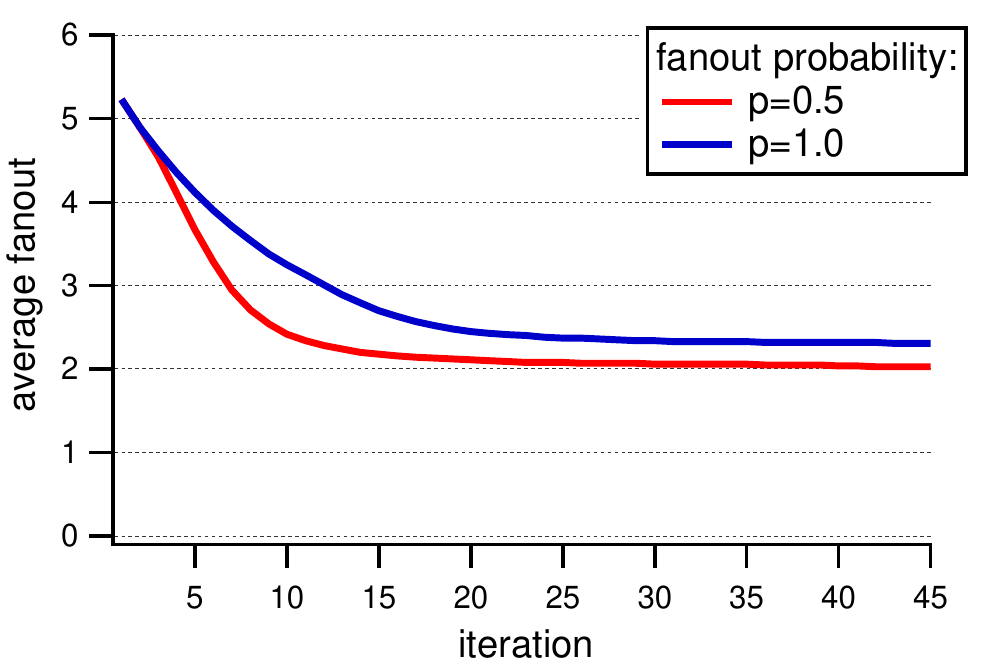}
            \caption{progress of fanout}
        \end{subfigure}
        \hfill
        \begin{subfigure}[t]{0.49\textwidth}
            \includegraphics[width=\columnwidth]{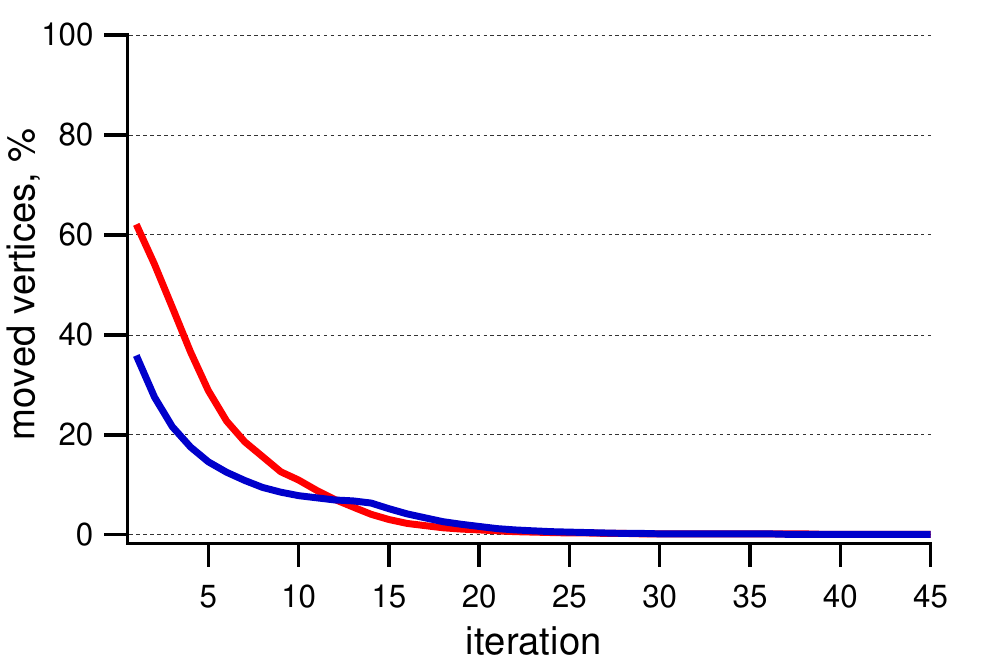}
            \caption{percentage of moved vertices per iteration}
        \end{subfigure}
    \end{minipage}
    \caption{Progress of fanout optimization with \texttt{SHP-k} for $p=0.5$ and $p=1.0$ on \texttt{soc-LJ} using $k=8$ buckets.}
    \label{fig:converge}
\end{figure}

One explanation, as mentioned in Section~\ref{sect:algo}, is that the local search is more likely to land in a local minimum with $p=1$. This is illustrated in Figure~\ref{fig:converge} for \texttt{SHP-k}, where the number of moved vertices per iteration on \texttt{soc-LJ} is significantly lower for $p=1$ than for $p=0.5$ at earlier iterations. The number of moved vertices for $p=0.5$ is below $0.1\%$ after iteration $35$; this number is below $0.01\%$ after iteration $49$. 
This behavior on \texttt{soc-Pokec} and \texttt{soc-LJ} for \texttt{SHP} was observed typical across many hypergraphs, and motivates our default parameters.  We set $p=0.5$ as the default for $\pfanout$, $60$ for the maximum number of refinement iterations of \texttt{SHP-k}, and $20$ iterations per bisection for \texttt{SHP-2}.

Figure~\ref{fig:fanout_prob} quantifies the impact of using probabilistic fanout for \texttt{SHP-2} across a variety of hypergraphs for $2$, $8$, and $32$ buckets. The left-hand plot displays the substantial percentage increases in fanout caused by using direct fanout optimization over the $p=0.5$ probabilistic fanout optimization.  Similar behavior is seen with \texttt{SHP-k}, and these results demonstrate the importance of using probabilistic fanout for \texttt{SHP}.  On average across these hypergraphs, direct fanout optimization would produce fanout values $45\%$ larger than probabilistic fanout with $p=0.5$.

\begin{figure}[!t] 
    \centering
    \begin{minipage}[b]{0.98\textwidth}    
        \begin{subfigure}[t]{0.49\textwidth}
            \includegraphics[width=\columnwidth]{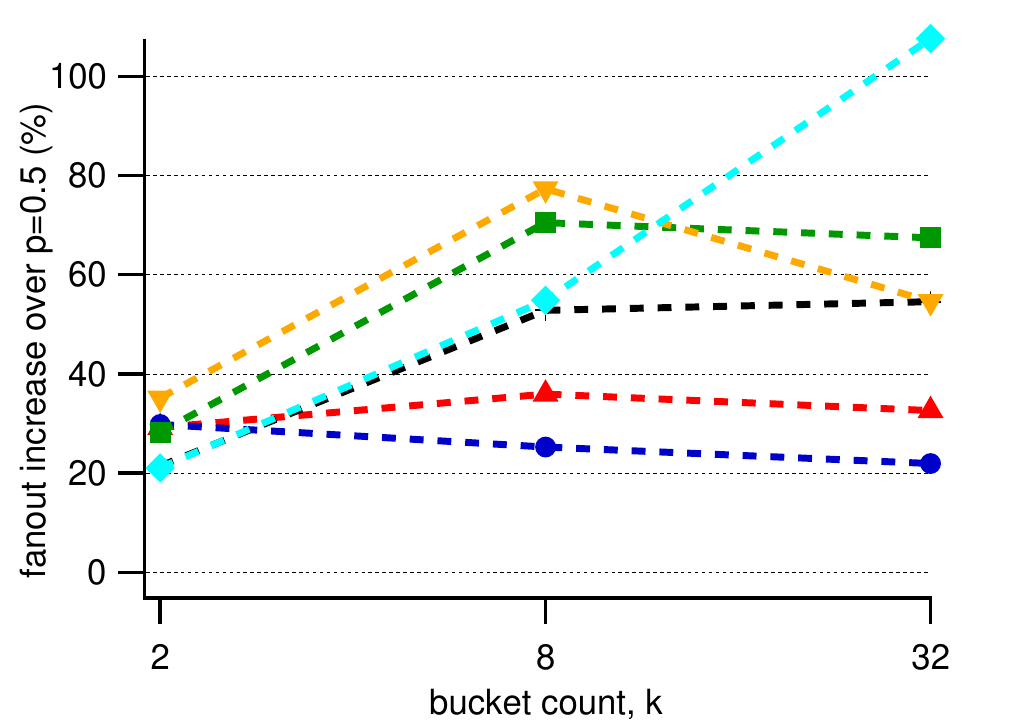}
            \caption{Direct fanout optimization, $p=1.0$}
        \end{subfigure}
        \hfill
        \begin{subfigure}[t]{0.49\textwidth}
            \includegraphics[width=\columnwidth]{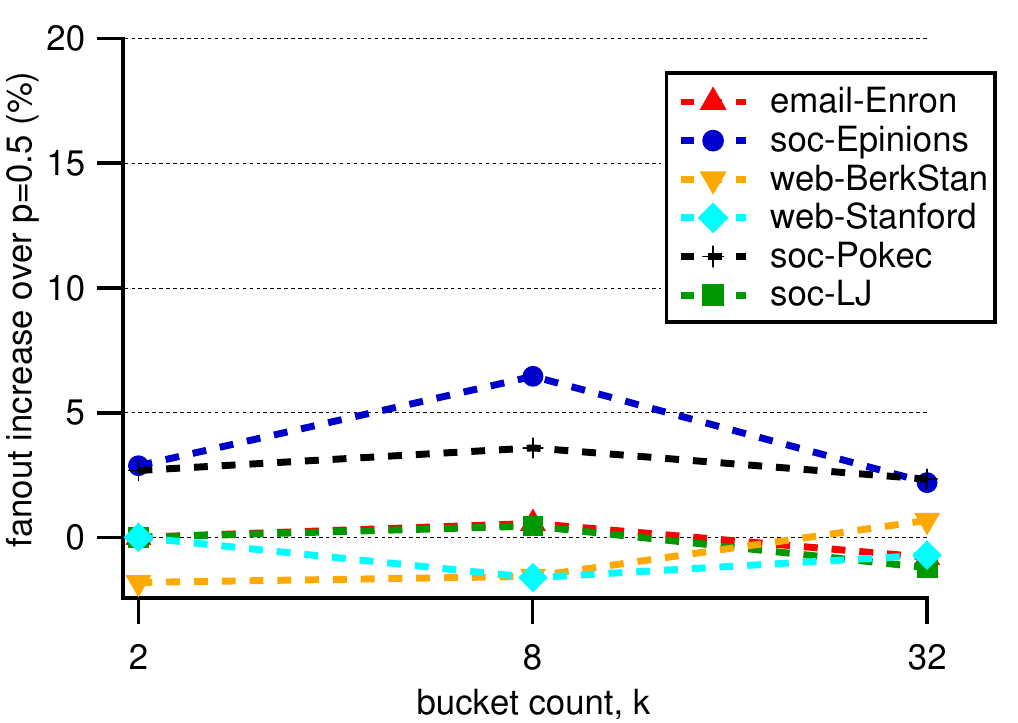}
            \caption{The clique-net optimization, $p=0.0$}
        \end{subfigure}
    \end{minipage}
    \caption{The comparison of different objective functions of \texttt{SHP-2} for $k \in \{2, 8, 32\}$:
        (a)~$0.5$-fanout vs $1$-fanout (direct) optimization, (b)~$0.5$-fanout vs $0$-fanout (clique-net) optimization.}
    \label{fig:fanout_prob}
\end{figure}

The right-hand plot of Figure~\ref{fig:fanout_prob} compares $\pfanout$ with the 
clique-net model defined in Section~\ref{sect:algo}, where 
the clique-net objective is optimized using \texttt{SHP-2}. 
The comparison to $p = 0.5$ reveals that clique-net optimization is often worse, but typically similar, depending on the graph and the number of buckets.  
In practice, we suggest optimizing for both $\pfanout$ and clique-net, as which surrogate objective performs better for fanout minimization depends on the graph.

\section{Discussion}

Storage sharding for production systems has many additional practical challenges~\cite{SH16}. Two requirements that arise from these challenges are (i)~\textit{incremental updates} of an existing partition and (ii)~\textit{balance across multiple dimensions}.

\begin{itemize}[(i)]
\item Incremental updates can be needed to avoid overhead from substantially changing a partition.  Our algorithm simply adapts to incremental updates by initializing with a previous partition and running a local search.  If a limited search moves too many data vertices, we can modify the move gain calculation to punish movement from the existing partition or artificially lower the movement probabilities returned via master in superstep four.

\item Basic $k$-way hypergraph partitioning balances the number of data vertices per bucket. Trivially, we can consider weighted data vertices, but additionally, a data vertex might have multiple associated dimensions (e.g., CPU cost, memory, disk resources etc.) that each require balance.  In practice, we have found that requiring strict balance on many dimensions substantially harms solution quality. Instead, we favor a simple heuristic that produces $c \cdot k$ buckets for some $c > 1$ that have loose balance requirements on all but one dimension, and merges them into $k$ buckets to satisfy load balance across all dimensions.
\end{itemize}

We stress that the storage sharding problem might have additional requirements that are not captured by our model. For example,
one could introduce some replication by allowing data records to be present on multiple servers, as it is done in~\cite{CJZM10,KQDK14}.
However, that would bring in new operational complications (e.g., synchronization between the servers when a data record is modified),
which are not always possible in an existing framework. Another potential extension of our model is 
to consider a better optimization goal that more accurately captures the relation between query latency and distribution
of data records among servers. We recently observed that the query fanout is not the only objective affecting its latency; 
the size of a request to a server also plays a role. For example, a query with $\fanout=2$ that needs $100$ data records
can be answered faster if the two servers contain an even number of records, $50$ and $50$, in comparison with a 
distribution of $99$ and $1$.
We leave a deeper investigation of this and other potential extensions of the model as an interesting future direction.

\section{Conclusion and Future Work}
In this paper, we presented Social Hash Partitioner, \texttt{SHP}, a distributed hypergraph partitioner that can optimize $\pfanout$, as well as the clique-net objective among others, through local search, and scales to far larger hypergraphs than existing hypergraph partitioning packages.  Because the careful design of \texttt{SHP} limits space, computational, and communication complexity, the applicability of our implementation to a hypergraph is only constrained by the number of Giraph machines available. We regularly apply \texttt{SHP} to hypergraphs that contain billions of vertices and hundreds of millions to even billions of hyperedges.  These hypergraphs correspond to bipartite graphs with billions of vertices and trillions of edges.

Despite the improved scalability and simplicity of the algorithm, our experiments demonstrate that \texttt{SHP} achieves comparable solutions to both single-machine and distributed hypergraph partitioners.  We note these results and conclusions might be different for other input hypergraphs (e.g. matrices from scientific computing, planar networks or meshes, etc.), and in cases where the hypergraph is small enough and solution quality is essential, running all available tools is recommended. While \texttt{SHP} occasionally produced the best solution in our experiments, other packages, especially \texttt{Zoltan} and \texttt{Mondriaan}, often claimed the lowest fanout.

Although our motivation for designing \texttt{SHP} is practical, it would be interesting to study fanout minimization from a theoretical point of view.  For example, the classical
balanced partitioning problem on unipartite graphs admits an $\Oh(\sqrt{\log k \log n})$-approximation for $\eps=2$~\cite{KNS09}
but it is unclear if a similar bound holds for our problem.  Alternatively, it would be interesting to know whether there is an optimal algorithm for some classes of hypergraphs, or an algorithm that provably finds a correct solution for certain random hypergraphs (e.g., generated with a planted partition model). Finally, it would be interesting to understand when and how minimizing $\pfanout$ speeds up algorithm convergence and improves solution quality over direct fanout minimization.

\section{Additional Authors}
Yaroslav Akhremtsev (Karlsruhe Institute of Technology) and  Alessandro Presta (Google) -- work done while at Facebook.

\section{Acknowledgments}
We thank Herald Kllapi for contributing to an earlier graph partitioner that evolved into \texttt{SHP}. We also thank Michael Stumm for providing comments on this work.

\newpage
\bibliographystyle{abbrv}
\bibliography{refs}

\begin{thebibliography}{10}

\bibitem{giraph}
{A}pache {G}iraph.
\newblock \url{http://giraph.apache.org/}.

\bibitem{giraphgit}
{S}ocial {H}ash {P}artitioner.
\newblock \url{https://issues.apache.org/jira/browse/GIRAPH-1131}.

\bibitem{snap}
Stanford large network dataset collection.
\newblock \url{https://snap.stanford.edu/data}.

\bibitem{alpert1996hybrid}
C.~J. Alpert, L.~W. Hagen, and A.~B. Kahng.
\newblock A hybrid multilevel/genetic approach for circuit partitioning.
\newblock In {\em Asia Pacific Conference on Circuits and Systems}, pages
  298--301. IEEE, 1996.

\bibitem{alpert1995recent}
C.~J. Alpert and A.~B. Kahng.
\newblock Recent directions in netlist partitioning: {A} survey.
\newblock {\em Integration, the VLSI journal}, 19(1-2):1--81, 1995.

\bibitem{AR06}
K.~Andreev and H.~R\"{a}cke.
\newblock Balanced graph partitioning.
\newblock {\em Theory of Computing Systems}, 39(6):929--939, 2006.

\bibitem{BMSW13}
D.~A. Bader, H.~Meyerhenke, P.~Sanders, and D.~Wagner.
\newblock Graph partitioning and graph clustering, 10th {DIMACS} implementation
  challenge workshop.
\newblock {\em Contemporary Mathematics}, 588, 2013.

\bibitem{BS13}
C.-E. Bichot and P.~Siarry.
\newblock {\em Graph partitioning}.
\newblock John Wiley \& Sons, 2013.

\bibitem{BMSSS16}
A.~Bulu{\c{c}}, H.~Meyerhenke, I.~Safro, P.~Sanders, and C.~Schulz.
\newblock Recent advances in graph partitioning.
\newblock In {\em Algorithm Engineering}, pages 117--158. Springer, 2016.

\bibitem{CA99}
{\"U}.~V. {\c{C}}ataly{\"u}rek and C.~Aykanat.
\newblock Hypergraph-partitioning-based decomposition for parallel
  sparse-matrix vector multiplication.
\newblock {\em IEEE Transactions on Parallel and Distributed Systems},
  10(7):673--693, 1999.

\bibitem{CJZM10}
C.~Curino, E.~Jones, Y.~Zhang, and S.~Madden.
\newblock Schism: a workload-driven approach to database replication and
  partitioning.
\newblock {\em VLDB Endowment}, 3(1-2):48--57, 2010.

\bibitem{taleAtScale}
J.~Dean and L.~A. Barroso.
\newblock The tail at scale.
\newblock {\em Communications of the ACM}, 56:74--80, 2013.

\bibitem{DKUC15}
M.~Deveci, K.~Kaya, B.~U{\c{c}}ar, and {\"U}.~V. {\c{C}}ataly{\"u}rek.
\newblock Hypergraph partitioning for multiple communication cost metrics:
  {M}odel and methods.
\newblock {\em Journal of Parallel and Distributed Computing}, 77:69--83, 2015.

\bibitem{DBHBC06}
K.~D. Devine, E.~G. Boman, R.~T. Heaphy, R.~H. Bisseling, and {\"U}.~V.
  {\c{C}}ataly{\"u}rek.
\newblock Parallel hypergraph partitioning for scientific computing.
\newblock In {\em International Parallel \& Distributed Processing Symposium},
  pages 10--pp. IEEE, 2006.

\bibitem{DKKOPS16}
L.~Dhulipala, I.~Kabiljo, B.~Karrer, G.~Ottaviano, S.~Pupyrev, and A.~Shalita.
\newblock Compressing graphs and indexes with recursive graph bisection.
\newblock In {\em International Conference on Knowledge Discovery and Data
  Mining}, pages 1535--1544. ACM, 2016.

\bibitem{ELWCK16}
S.~Edunov, D.~Logothetis, C.~Wang, A.~Ching, and M.~Kabiljo.
\newblock Darwini: Generating realistic large-scale social graphs.
\newblock {\em arXiv:1610.00664}, 2016.

\bibitem{FF15}
A.~E. Feldmann and L.~Foschini.
\newblock Balanced partitions of trees and applications.
\newblock {\em Algorithmica}, 71(2):354--376, 2015.

\bibitem{FM82}
C.~M. Fiduccia and R.~M. Mattheyses.
\newblock A linear-time heuristic for improving network partitions.
\newblock In {\em 19th Conference on Design Automation}, pages 175--181. IEEE,
  1982.

\bibitem{GHKS14}
L.~Golab, M.~Hadjieleftheriou, H.~Karloff, and B.~Saha.
\newblock Distributed data placement to minimize communication costs via graph
  partitioning.
\newblock In {\em International Conference on Scientific and Statistical
  Database Management}, pages 20:1--20:12. ACM, 2014.

\bibitem{KAKS99}
G.~Karypis, R.~Aggarwal, V.~Kumar, and S.~Shekhar.
\newblock Multilevel hypergraph partitioning: applications in {VLSI} domain.
\newblock {\em IEEE Transactions on Very Large Scale Integration Systems},
  7(1):69--79, 1999.

\bibitem{KK95}
G.~Karypis and V.~Kumar.
\newblock {METIS}--unstructured graph partitioning and sparse matrix ordering
  system, version 2.0.
\newblock 1995.

\bibitem{KK00}
G.~Karypis and V.~Kumar.
\newblock Multilevel k-way hypergraph partitioning.
\newblock {\em VLSI design}, 11(3):285--300, 2000.

\bibitem{KL70}
B.~W. Kernighan and S.~Lin.
\newblock An efficient heuristic procedure for partitioning graphs.
\newblock {\em Bell system technical journal}, 49(2):291--307, 1970.

\bibitem{Kie16}
T.~Kiefer.
\newblock {\em Allocation Strategies for Data-Oriented Architectures}.
\newblock PhD thesis, Dresden, Technische Universit{\"a}t Dresden, 2016.

\bibitem{KNS09}
R.~Krauthgamer, J.~S. Naor, and R.~Schwartz.
\newblock Partitioning graphs into balanced components.
\newblock In {\em Symposium on Discrete Algorithms}, pages 942--949. SIAM,
  2009.

\bibitem{KQDK14}
K.~A. Kumar, A.~Quamar, A.~Deshpande, and S.~Khuller.
\newblock {SWORD}: workload-aware data placement and replica selection for
  cloud data management systems.
\newblock {\em The VLDB Journal}, 23(6):845--870, 2014.

\bibitem{LMT90}
T.~Leighton, F.~Makedon, and S.~Tragoudas.
\newblock Approximation algorithms for {VLSI} partition problems.
\newblock In {\em Circuits and Systems}, pages 2865--2868. IEEE, 1990.

\bibitem{STA12}
R.~O. Selvitopi, A.~Turk, and C.~Aykanat.
\newblock Replicated partitioning for undirected hypergraphs.
\newblock {\em Journal of Parallel and Distributed Computing}, 72(4):547--563,
  2012.

\bibitem{SH16}
A.~Shalita, B.~Karrer, I.~Kabiljo, A.~Sharma, A.~Presta, A.~Adcock, H.~Kllapi,
  and M.~Stumm.
\newblock {S}ocial {H}ash: {A}n assignment framework for optimizing distributed
  systems operations on social networks.
\newblock In {\em 13th Usenix Conference on Networked Systems Design and
  Implementation}, pages 455--468, 2016.

\bibitem{ST97}
H.~D. Simon and S.-H. Teng.
\newblock How good is recursive bisection?
\newblock {\em Journal on Scientific Computing}, 18(5):1436--1445, 1997.

\bibitem{TK08}
A.~Trifunovi{\'c} and W.~J. Knottenbelt.
\newblock Parallel multilevel algorithms for hypergraph partitioning.
\newblock {\em Journal of Parallel and Distributed Computing}, 68(5):563--581,
  2008.

\bibitem{VB05}
B.~Vastenhouw and R.~H. Bisseling.
\newblock A two-dimensional data distribution method for parallel sparse
  matrix-vector multiplication.
\newblock {\em SIAM Review}, 47(1):67--95, 2005.

\bibitem{YWMW16}
W.~Yang, G.~Wang, L.~Ma, and S.~Wu.
\newblock A distributed algorithm for balanced hypergraph partitioning.
\newblock In {\em Advances in Services Computing}, pages 477--490. Springer,
  2016.

\end{thebibliography}

\end{document}